\documentclass[11pt]{article}
\usepackage[left=.8in, right=.8in, top=.9in, bottom=.9in]{geometry}
\usepackage{amssymb}
\usepackage{amsmath} 
\usepackage{amsthm}
\usepackage{verbatim}
\usepackage{mathtools}
\usepackage{enumerate}
\usepackage{graphicx}
\usepackage{authblk}
\usepackage{cite}
\usepackage{thmtools}
\usepackage{thm-restate}
\usepackage{hyperref}
\usepackage{cleveref}

\hypersetup{
	colorlinks=true,
	citecolor=blue,
	linkcolor=blue,
	urlcolor=black
}

\DeclareMathOperator{\mat}{Mat}
\DeclareMathOperator{\tr}{tr}

\DeclareMathOperator{\Umath}{U}
\DeclareMathOperator{\herm}{Herm}
\DeclareMathOperator{\pos}{Pos}
\DeclareMathOperator{\proj}{Proj}
\DeclareMathOperator{\den}{Den}
\DeclareMathOperator{\T}{T}
\DeclareMathOperator{\chan}{Chan}

\DeclarePairedDelimiter\bra{\langle}{\rvert}
\DeclarePairedDelimiter\ket{\lvert}{\rangle}
\DeclarePairedDelimiterX\braket[2]{\langle}{\rangle}{#1 \delimsize\vert #2}

\DeclareMathAlphabet{\mathbbm}{U}{bbold}{m}{n}

\theoremstyle{plain}
\newtheorem{theorem}{Theorem}
\newtheorem{corollary}[theorem]{Corollary}
\newtheorem{lemma}[theorem]{Lemma}
\newtheorem{proposition}[theorem]{Proposition}

\theoremstyle{remark}
\newtheorem*{remark}{Remark}

\theoremstyle{definition}
\newtheorem{definition}[theorem]{Definition}

\theoremstyle{definition}
\newtheorem{openProb}{Open Problem}

\title{Eliminating Intermediate Measurements in Space-Bounded Quantum Computation}
\author{Bill Fefferman}
\author{Zachary Remscrim\thanks{Corresponding author, remscrim@uchicago.edu; portions of this research were completed while a member of the Department of Mathematics at MIT.}}
\affil{Department of Computer Science, The University of Chicago}
\date{}

\begin{document}
	
	\maketitle
	
	\begin{abstract}
		
		A foundational result in the theory of quantum computation, known as the ``principle of safe storage,'' shows that it is always possible to take a quantum circuit and produce an equivalent circuit that makes all measurements at the end of the computation.  While this procedure is time efficient, meaning that it does not introduce a large overhead in the number of gates, it uses extra ancillary qubits, and so is not generally space efficient.  It is quite natural to ask whether it is possible to eliminate intermediate measurements without increasing the number of ancillary qubits.
		
		We give an affirmative answer to this question by exhibiting a procedure to eliminate all intermediate measurements that is simultaneously space efficient and time efficient. In particular, this shows that the definition of a space-bounded quantum complexity class is robust to allowing or forbidding intermediate measurements.  A key component of our approach, which may be of independent interest, involves showing that the well-conditioned versions of many standard linear-algebraic problems may be solved by a quantum computer in less space than seems possible by a classical computer.
	\end{abstract}
	
	\section{Introduction}\label{sec:intro}
	Quantum computation has the potential to obtain dramatic speedups for important problems such as quantum simulation (see, e.g., \cite{feynman,lloyd}) and integer factorization \cite{shor1994algorithms}. While fault-tolerant, fully scalable quantum computers may still be far from fruition, we have now entered an exciting period in which impressive but resource constrained quantum experiments are being implemented in many academic and industrial labs.  As the field transitions from ``proof of principle'' demonstrations of provable quantum advantage to solving useful problems on near-term experiments, it is particularly critical to characterize the algorithmic power of feasible models of quantum computations that have restrictive resources, such as ``time'' (i.e., the number of gates in the circuit) and ``space'' (i.e., the number of qubits on which the circuit operates), and to understand how these resources can be traded-off.
	
	A foundational question in this area asks if it is possible to \textit{space-efficiently} eliminate intermediate measurements in a quantum computation (see e.g.,  \cite{watrous1999space,watrous2001quantum,watrous2003complexity,melkebeek2012time,ta2013inverting,fefferman2018complete,jozsa2010matchgate,perdrix2006classically}).  While a classic result known as the ``principle of safe storage''\footnote{The ``principle of deferred measurement'' is another common name for this result.} states that it is always possible to \textit{time-efficiently} defer intermediate measurements to the end of a computation \cite{aharonov1998quantum, nielsen2002quantum}, this procedure uses extra ancilla qubits, and so is not generally space efficient.  More specifically, if a quantum circuit $Q$ acts on $s$ qubits and performs $m$ intermediate measurements, the circuit $Q'$ constructed using this principle operates on $s+poly(m)$ qubits; if, for example, $s=O(\log t)$ and $m=\Theta(t)$, this entails an \textit{exponential blowup} in the amount of needed \textit{space}.
	
	Our main result solves this problem.  We show that every problem solvable by a ``general" quantum algorithm, which may make arbitrary use of quantum measurements, can also be solved, using the same amount of space, by a ``unitary'' quantum algorithm, which may not perform any intermediate measurements. As an immediate corollary, this shows that, in the space-bounded setting, unitary quantum algorithms are at least as powerful as probabilistic algorithms, resolving a longstanding open problem \cite{watrous2001quantum,melkebeek2012time}. 
	
	In the process of proving this result, we also obtain the following result, which is likely of independent interest: approximating the solution of the ``well-conditioned" versions of various standard linear-algebraic problems, such as the determinant problem, the matrix inversion problem, or the matrix powering problem, is complete for the class of bounded-error logspace quantum algorithms. These are a new class of natural problems on which quantum computers seem to outperform their classical counterparts.
	
	\subsection{Eliminating Intermediate Measurements}\label{sec:intro:elimMeasurements}
	
	Before proceeding further, it is worthwhile to briefly discuss why it is desirable to be able to eliminate intermediate measurements. Firstly, quantum measurements are a natural resource, much as time and space are. In addition to the general desirability of using as few resources as possible in any sort of computational task, it is especially desirable to avoid intermediate measurements, due to the technical challenges involved in implementing such measurements and resetting qubits to their initial states (for a discussion of these issues from an experimental perspective see, e.g., \cite{divincenzo}). Secondly, unitary computations are \textit{reversible}, whereas quantum measurement is an inherently irreversible process. The ability to ``undo" a unitary subroutine, by running it in reverse, is routinely used in the design and analysis of quantum algorithms (see, for instance, \cite{bennett1997strengths,marriott2005quantum,watrous2009zero,fefferman2016space,fefferman2018complete,shor2008estimating,nagaj2009fast}). Moreover, reversible computations may be performed without generating heat \cite{landauer1961irreversibility}. Thirdly, by demonstrating that unitary quantum space and general quantum space are equivalent in power, we show that the definition of quantum space is quite \textit{robust}. Allowing intermediate measurements, or even general quantum operations, does not provide any additional power in the space-bounded setting.
	
	Let $\mathsf{BQ_U SPACE}(s)$ (resp. $\mathsf{BQSPACE}(s)$) denote the class of promise problems recognizable with two-sided bounded-error by a uniform family of unitary (resp. general) quantum circuits, where, for each input of length $n$, there is a corresponding circuit that operates on $O(s(n))$ qubits and has $2^{O(s(n))}$ gates. Note that it is standard to require that the running time of a computation is at most exponential in its space bound; see, for instance, \cite{watrous1999space,watrous2003complexity,melkebeek2012time,saks1996randomization} for the importance of this restriction in quantum and/or probabilistic space-bounded computation. Furthermore, let $\mathsf{Q_UMASPACE}(s)$ (resp. $\mathsf{QMASPACE}(s)$) denote those promise problems recognized by a unitary (resp. general) quantum Merlin-Arthur protocol that operates in space $O(s(n))$ and time $2^{O(s(n))}$. An equivalent definition of these complexity classes may be given using quantum Turing machines; see \Cref{sec:prelim:quantumSpace} for further details. 
	
	Our main result is:
	
	\begin{restatable}{theorem}{restateBquspaceEqualsBqspace}\label{thm:intro:bquspaceEqualsBqspace}
		For any space-constructible function $s: \mathbb{N} \rightarrow \mathbb{N}$, where $s(n)=\Omega(\log n)$, we have $$\mathsf{BQ_U SPACE}(s)=\mathsf{BQSPACE}(s)=\mathsf{Q_UMASPACE}(s)=\mathsf{QMASPACE}(s).$$ 
	\end{restatable}
	
	\begin{remark}
		Note that $\mathsf{BPSPACE}(s) \subseteq \mathsf{BQ_U SPACE}(s)$ was not previously known to hold, where $\mathsf{BPSPACE}(s)$ denotes the class of problems recognizable by a \textit{probabilistic} algorithm in space $O(s)$ (and time $2^{O(s)}$); see, e.g., \cite{watrous2001quantum,melkebeek2012time} for discussion of this question. As one would expect quantum computation to generalize probabilistic computation, the lack of a proof of this containment was unfortunate. Since it is clear that $\mathsf{BPSPACE}(s) \subseteq \mathsf{BQSPACE}(s)$, we have $\mathsf{BPSPACE}(s) \subseteq \mathsf{BQ_U SPACE}(s)$, resolving this question.
	\end{remark}  
	
	\begin{remark}
		To further clarify the parameters of our result, given a \textit{general} quantum algorithm that operates in space $s$ and time $t$, we produce a \textit{unitary} quantum algorithm that operates in space $O(s +\log t)$ and time $poly(t 2^{s})$. Note that these parameters coincide with that of the space-efficient simulation of a deterministic algorithm by a (classical) reversible algorithm \cite{lange2000reversible}. Further note that in the extreme (but natural) setting in which $t=2^{O(s)}$ (e.g. quantum logspace), this procedure is simultaneously space-efficient and time-efficient. On the other hand, in the opposite extreme setting in which $t=poly(s)$, this procedure is no longer time-efficient; of course, in this setting, the standard ``principle of deferred measurement" is simultaneously space-efficient and time-efficient. Between these two extremes, one has time-space tradeoffs analogous to those of the simulation of deterministic algorithms by reversible algorithms \cite{buhrman2001time}. 
	\end{remark}  
	
	We also study the one-sided (bounded-error and unbounded-error) analogues of the aforementioned two-sided bounded-error space-bounded quantum complexity classes. We establish analogous results concerning the relationship between the unitary and general versions of these classes; see \Cref{sec:singular} for a formal statement of these results.
	
	\subsection{Exact and Approximate Linear Algebra}\label{sec:intro:linAlg}
	
	Let $\mathsf{intDET}$ denote the problem of computing the determinant of an $n \times n$ integer-valued matrix, and, following its original definition by Cook \cite{cook1985taxonomy}, let $\mathsf{DET}^*$ denote the class of problems $\mathsf{NC}^1$ (Turing) reducible to $\mathsf{intDET}$. Let $\mathsf{BQ_U L}=\mathsf{BQ_U SPACE}(\log(n))$, $\mathsf{BQL}=\mathsf{BQSPACE}(\log(n))$, and $\mathsf{BPL}=\mathsf{BPSPACE}(\log(n))$ denote the bounded-error quantum and probabilistic logspace classes. Before our work, the following relationships were known \cite{watrous1999space,watrous2003complexity}: $\mathsf{BQ_U L} \subseteq \mathsf{BQL} \subseteq \mathsf{DET}^*$ and $\mathsf{BPL} \subseteq \mathsf{BQL}$. Many natural linear-algebraic problems are $\mathsf{DET}^*$-complete, such as $\mathsf{intDET}$, $\mathsf{intMATINV}$ (the problem of computing a single entry of the inverse of a matrix), and $\mathsf{intITMATPROD}$ (the problem of computing a single entry of the product of polynomially-many matrices). It seems rather unlikely that any such $\mathsf{DET}^*$-complete problem is in the class $\mathsf{BQL}$, as this would imply $\mathsf{BQL}=\mathsf{DET}^*$, and, therefore, $\mathsf{NL} \subseteq \mathsf{BQL}$. 
	
	We next consider the problem of \textit{approximating} the answer to such linear-algebraic problems. Let $poly\text{-conditioned-}\mathsf{MATINV}$ denote the problem of approximating, to additive $1/poly(n)$ accuracy, a single entry of the inverse of an $n \times n$ matrix $A$ with \textit{condition number} $\kappa(A)=poly(n)$ (see \Cref{sec:determinant} for a precise definition). Ta-Shma \cite{ta2013inverting}, building on the landmark result of Harrow, Hassidim, and Lloyd \cite{harrow2009quantum}, showed $poly\text{-conditioned-}\mathsf{MATINV} \in \mathsf{BQL}$. Fefferman and Lin \cite{fefferman2018complete} improved upon this result by presenting a \textit{unitary} quantum logspace algorithm and proving a matching $\mathsf{BQ_U L}$-hardness result, thereby exhibiting the first known natural $\mathsf{BQ_U L}$-complete (promise) problem. We further extend this line of research by proving the following theorem, which demonstrates an intriguing relationship between $\mathsf{BQ_U L}$ and $\mathsf{DET}^*$: 
	
	\begin{restatable}{theorem}{restateBqulComplete}\label{thm:bqulComplete}
		All of the $poly$-conditioned versions of the ``standard" $\mathsf{DET}^*$-complete problems, given in \Cref{def:detPromiseProblems:detAndInv,def:detPromiseProblems:powAndItprod} are $\mathsf{BQ_U L}$-complete.
	\end{restatable}
	
	This shows that several natural linear-algebraic problems are in $\mathsf{BQ_U L}$, and, moreover, are not in $\mathsf{BPL}$ (unless $\mathsf{BQ_U L}=\mathsf{BPL}$). In particular, the above theorem shows $poly\text{-conditioned-}\mathsf{ITMATPROD} \in \mathsf{BQ_U L}$. We also show that this problem is $\mathsf{BQL}$-hard, which implies $\mathsf{BQL}=\mathsf{BQ_U L}$; \Cref{thm:intro:bquspaceEqualsBqspace}, which states the more general equivalence for any larger space bound, then follows from a standard padding argument.
	
	We next exhibit several other applications of this theorem. Firstly, in \Cref{sec:determinant:approximateCounting}, we consider \textit{fully logarithmic approximation schemes}, whose study was initiated by Doron and Ta-Shma \cite{doron2015randomization}. Using the preceding theorem, we show that the $\mathsf{BQL}$ vs. $\mathsf{BPL}$ question is equivalent to several distinct questions involving the relative power of quantum and probabilistic fully logarithmic approximation schemes.	Secondly, consider $\kappa(n)\text{-conditioned-}\mathsf{DET}$, the problem of approximating, to within a \textit{multiplicative} factor $1+1/poly(n)$, the determinant of an $n \times n$ matrix with condition number $\kappa(n)$. Boix-Adser\`{a}, Eldar, and Mehraban \cite{boix2019approximating} recently showed that $\kappa(n)\text{-conditioned-}\mathsf{DET}\in \mathsf{DSPACE}(\log(n) \log(\kappa(n))  poly(\log\log n))$. They also raised the following question: is $poly\text{-conditioned-}\mathsf{DET}$  $\mathsf{BQL}$-complete? As an immediate consequence of \Cref{thm:bqulComplete}, we answer their question in the affirmative.
	
	\begin{corollary}\label{prop:introPolyCondDetIsBQULcomplete}
		$poly\text{-conditioned-}\mathsf{DET}$ is $\mathsf{BQL}(=\mathsf{BQ_U L})$-complete.
	\end{corollary}
	
	To see the significance of the previous corollary, recall the well-known ``dequantumization" result given by Watrous \cite{watrous2003complexity}: $\mathsf{BQL} \subseteq \mathsf{DSPACE}(\log^2 n)$. It is natural to ask if a stronger upper bound on $\mathsf{BQL}$ can be established. We note that the strongest currently known ``derandomization" result of this type, given by Saks and Zhou \cite{saks1999bphspace}, states $\mathsf{BPL} \subseteq \mathsf{DSPACE}(\log^{\frac{3}{2}} n)$. Note that the statement $\mathsf{BQL} \subseteq \mathsf{DSPACE}(\log^{2-\epsilon} n)$ would follow from either a small improvement in the result of Boix-Adser\`{a}, Eldar, and Mehraban (i.e., proving a stronger upper bound on the needed amount of deterministic space in terms of $\kappa(n)$) or from a small improvement in our result (i.e., proving $\kappa(n)\text{-conditioned-}\mathsf{DET}$ remains $\mathsf{BQL}$-hard for \textit{subpolynomial} $\kappa(n)$). Therefore, if $\mathsf{BQL} \not \subseteq \mathsf{DSPACE}(\log^{2-\epsilon} n)$, $\forall \epsilon>0$, then both our result and their result are essentially optimal (in terms of the dependence on $\kappa(n)$).
	
	Moreover, we note that our paper and that of Boix-Adser\`{a}, Eldar, and Mehraban used similar power series techniques to produce space-efficient algorithms for $\kappa(n)\text{-conditioned-}\mathsf{DET}$. However, our quantum algorithm can make use of a power series with an \textit{exponentially} larger number of terms than seems possible for their (or any other) classical algorithm. This suggests a possible mechanism for explaining the supposed advantage of quantum computers over classical computers in the space-bounded setting.
	
	In \Cref{sec:singular}, we study well-conditioned versions of the ``standard" $\mathsf{C_{=}L}$-complete problems. We establish a result, very much analogous to \Cref{thm:bqulComplete}, which shows that these problems are complete for the \textit{one-sided} error versions of space-bounded quantum complexity classes. This enables us to prove results, analogous to \Cref{thm:intro:bquspaceEqualsBqspace}, concerning the relative power of unitary and general quantum space in the one-sided error cases. We conclude by stating several open problems related to our work in \Cref{sec:discussion}. 
	
	\subsection{Techniques}\label{sec:intro:techniques}
	
	We now briefly discuss the techniques used to prove \Cref{thm:bqulComplete}, which states that the $poly$-conditioned versions of the ``standard" $\mathsf{DET}^*$-complete problems are $\mathsf{BQ_U L}$-complete. As discussed earlier, Fefferman and Lin \cite{fefferman2018complete} showed that $poly\text{-conditioned-}\mathsf{MATINV}$ is $\mathsf{BQ_U L}$-complete. In order to establish the $\mathsf{BQ_U L}$-completeness of the other $poly$-conditioned problems, we exhibit a long cycle of reductions through these problems. We note that reductions between the standard versions of these problems (i.e., where there is no assumption of being well-conditioned) are well-known \cite{cook1985taxonomy,allender1996relationships,valiant1992boolean,vinay1991counting,toda1991counting,mahajan1997determinant,damm1991det,berkowitz1984computing}. However, these reductions, generally, \textit{do not} preserve the property of being $poly$-conditioned. Therefore, we must exhibit reductions that are rather different from the ``standard" reductions. 
	
	As a motivating example, consider $poly\text{-conditioned-}\mathsf{DET}^{+}$ and $poly\text{-conditioned-}\mathsf{SUMITMATPROD}$. While Berkowitz's algorithm \cite{berkowitz1984computing} provides a reduction from $\mathsf{DET}^{+}$ to $\mathsf{SUMITMATPROD}$, this reduction does not preserve the property of being $poly$-conditioned. We now provide a brief sketch of a reduction which does preserve this property; see \Cref{lemma:reduction:posDetToSumIterMatProd} for a formal proof. Consider a positive definite $n \times n$ matrix $H$, with $\sigma_1(H) \leq 1$ and $\kappa(H)=poly(n)$. We wish to obtain an additive $1/poly(n)$ approximation of $\ln(\det(H))$. By Jacobi's formula, $\ln(\det(H))=\tr(\ln(H))$, where $\ln(H)$ denotes the matrix logarithm. We have $\sigma_1(I-H) \leq 1-1/poly(n) <1$, which implies that the series $-\sum\limits_{k=1}^{\infty} \frac{(I-H)^k}{k}$ converges to $\ln(H)$. Therefore, $\ln(\det(H))=-\sum\limits_{k=1}^{\infty} \frac{\tr((I-H)^k)}{k}$. Moreover, as $\sigma_1(I-H) \leq 1-1/poly(n)$, the aforementioned series converges ``quickly," which implies that, for some $m=poly(n)$, the quantity $-\sum\limits_{k=1}^{m} \frac{\tr((I-H)^k)}{k}$ is a sufficiently good approximation of $\ln(\det(H))$. Estimating this quantity corresponds to an instance of $poly\text{-conditioned-}\mathsf{SUMITMATPROD}$.
	
	In \Cref{sec:determinant:bqulComplete}, we exhibit a collection of reductions between these various linear-algebraic problems, which use a variety of techniques to preserve the property of being $poly$-conditioned.

	\subsection{Related Work}\label{sec:intro:relatedWork}
	
	Simultaneously and independently of our work, Girish, Raz, and Zhan \cite{girish2020quantum} proved the following weaker version of our \Cref{thm:bqulComplete}: $contraction\text{-}\mathsf{MATPOW} \in \mathsf{BQ_U L}$, where $contraction\text{-}\mathsf{MATPOW}$ is a special case of our $poly\text{-conditioned-}\mathsf{MATPOW}$. We note that the techniques used in their proof differed substantially from ours. As a consequence of this result, they then obtain the following weaker version of our \Cref{thm:intro:bquspaceEqualsBqspace}: $\mathsf{BQ_U L}=\mathsf{BQ_Q L}$, where $\mathsf{BQ_Q L} \subseteq \mathsf{BQL}$ is a version of quantum logspace that allows a special type of intermediate measurements to be performed, but does not allow using the (classical) result of earlier measurements to control later steps of the computation.

	\section{Preliminaries}\label{sec:prelim}
	
	\subsection{General Notation and Definitions}\label{sec:prelim:notation}
	
	Let $\mat(n)=\mathbb{C}^{n \times n}$ denote the set of all $n \times n$ complex matrices and let $\herm(n)=\{A \in \mat(n):A=A^{\dagger}\}$ denote the set of all $n \times n$ Hermitian matrices. For $A \in \mat(n)$, let $\sigma_1(A) \geq \cdots \geq \sigma_n(A) \geq 0$ denote its singular values, and let $\lambda_1(A),\ldots,\lambda_n(A) \in \mathbb{C}$ denote its eigenvalues (with multiplicity); if $A \in \herm(n)$, then $\lambda_j(A) \in \mathbb{R}, \forall j$, and we order the eigenvalues such that $\lambda_1(A) \geq \cdots \geq \lambda_n(A)$. Let $I_n$ denote the $n \times n$ identity matrix, $\pos(n)=\{A \in \herm(n):\lambda_n(A) \geq 0\}$ denote the $n \times n$ positive semidefinite matrices, $\proj(n)=\{A \in \pos(n):A^2=A\}$ denote the $n \times n$ projection matrices, $\Umath(n)=\{A \in \mat(n):AA^{\dagger}=I_n\}$ denote the $n \times n$ unitary matrices, and $\den(n)=\{A \in \pos(V):\tr(A)=1\}$ denote the $n \times n$ density matrices. Let $\mathbb{Q}[i]_n=\{\frac{r+ci}{d}:r,c,d \in \mathbb{Z}, \lvert r \rvert, \lvert c \rvert, \lvert d \rvert \leq 2^{O(n)}\}$ denote the set of all $O(n)$-bit Gaussian rationals and let $\widehat{\mat}(n)$ (resp. $\widehat{\herm}(n),\widehat{\pos}(n),$ etc.) denote the subset of $\mat(n)$ (resp. $\herm(n),\pos(n),$ etc.) consisting of those matrices whose entries are all in $\mathbb{Q}[i]_n$. We define $\mat(n,c,d)=\{A \in \mat(n):c \leq \sigma_n(A) \leq \sigma_1(A) \leq d\}$ and we also analogously define $\widehat{\mat}(n,c,d)$, $\herm(n,c,d)$, etc. Let $[n]=\{1,\ldots,n\}$.
	
	We assume that the reader has familiarity with quantum computation and the theory of quantum information; see, for instance, \cite{nielsen2002quantum,watrous2018theory,kitaev1997quantum} for an introduction. A quantum system, on $s$ qubits, that is in a \textit{pure state} is described by a unit vector $\ket{\psi}$ in the $2^s$-dimensional Hilbert space $\mathbb{C}^{2^s}$. A \textit{mixed state} of the same system is described by some ensemble $\{(p_i,\ket{\psi_i}):i \in I \}$, for some index set $I$, where $p_i \in [0,1]$ denotes the probability of the system being in the pure state $\ket{\psi_i} \in\mathbb{C}^{2^s}$, and $\sum_i p_i=1$. This ensemble corresponds to the density matrix $A=\sum_i p_i \ket{\psi_i}\bra{\psi_i} \in \den(2^s)$. 
	
	Let $\T(n,m)$ denote the set of all superoperators of the form $\Phi:\mat(n) \rightarrow \mat(m)$ (i.e., $\Phi$ is a $\mathbb{C}$-linear map from the $\mathbb{C}$-vector space $\mat(n)$ to the $\mathbb{C}$-vector space $\mat(m)$). Let $\T(n)=\T(n,n)$ and let $\mathbbm{1}_{n} \in \T(n)$ denote the identity operator. Consider some $\Phi \in \T(n,m)$. We say that $\Phi$ is \textit{positive} if, $\forall A \in \pos(n)$, we have $\Phi(A) \in \pos(m)$. We say that $\Phi$ is \textit{completely positive} if $\Phi \otimes \mathbbm{1}_r$ is positive, $\forall r \in \mathbb{N}$, where $\otimes$ denotes the tensor product. We say that $\Phi$ is \textit{trace-preserving} if $\tr(\Phi(A))=\tr(A)$, $\forall A \in \mat(n)$. If $\Phi$ is both completely positive and trace-preserving, then we say $\Phi$ is a \textit{quantum channel}; let $\chan(n,m)=\{\Phi \in \T(n,m):\Phi \text{ is a quantum channel}\}$ denote the set of all such channels, and let $\chan(n)=\chan(n,n)$. Let $\text{vec}$ denote the usual vectorization map that takes a matrix $A \in \mat(n)$ to the vector $\text{vec}(A) \in \mathbb{C}^{n^2}$ consisting of the entries of $A$ (in some fixed order). For $\Phi \in \T(n)$, let $K(\Phi) \in \mat(n^2)$ denote the \textit{natural representation} of $\Phi$, which is defined to be the (unique) matrix for which $\text{vec}(\Phi(A))=K(\Phi)\text{vec}(A)$, $\forall A \in \mat(n)$. 
	
	We consider parameterized promise problems of the form $\mathcal{P}=(\mathcal{P}_{n,f_1,\ldots,f_h})_{n \in \mathbb{N}}$, for functions $f_1,\ldots,f_h:\mathbb{N} \rightarrow \mathbb{R}$; $\mathcal{P}_{n,f_1,\ldots,f_h}$ consists of instances of size $n$ which satisfy various conditions expressed in terms of $f_1(n),\ldots,f_h(n)$. For a promise problem $\mathcal{P}$ defined over some alphabet $\Sigma$, we, by slight abuse of notation, also write $\mathcal{P}$ to denote the subset of $\Sigma^*$ that satisfies the promise; analogously, we write $\mathcal{P}_{n,f_1,\ldots,f_h}$ to denote those instances of size $n$ that satisfy the promise. For $\langle X \rangle \in \mathcal{P}$, let $\mathcal{P}(\langle X \rangle)\in \{0,1\}$ denote the desired output on input $X$. We also use the notation $\mathcal{P}=(\mathcal{P}_1,\mathcal{P}_0)$, where $\mathcal{P}_j=\{\langle X \rangle \in \mathcal{P}:\mathcal{P}(\langle X \rangle)=j\}$. 	
	
	\subsection{Space-Bounded Quantum Computation}\label{sec:prelim:quantumSpace}
	
	We briefly recall the definitions of several needed space-bounded quantum complexity classes. We begin by noting that, in many of the previous papers that considered space-bounded quantum computation \cite{watrous1999space,watrous2001quantum,watrous2003complexity,watrous2009encyclopedia,ta2013inverting,melkebeek2012time,perdrix2006classically,jozsa2010matchgate}, \textit{quantum Turing machines} were used to define the various complexity classes of interest. Arguably, this is the ``natural" model to be used to define these classes. However, as the (equivalent) model of uniformly generated \textit{quantum circuits} are more familiar to quantum complexity theorists and physicists, we state these definitions using quantum circuits. We emphasize that all of the results in this paper apply to both quantum circuits and quantum TMs. 
	
	\begin{definition}\label{def:unitaryQuantumCircuit}
		A (unitary) \textit{quantum circuit} is a sequence of quantum gates, each of which is a member of some fixed set of gates that is universal for quantum computation (e.g., $\{\text{H,CNOT,T}\}$).  We say that a family of quantum circuits $\{Q_w:w \in \mathcal{P}\}$ is $\mathsf{DSPACE}(s)$-uniform if there is a deterministic TM that, on any input $w \in \mathcal{P}$, runs in space $O(s(\lvert w \rvert))$ (and hence time $2^{O(s(\lvert w \rvert))}$), and outputs a description of $Q_w$. 
	\end{definition}
	
	\begin{definition}\label{def:unitaryQuantumSpace}
		Consider functions $c,k:\mathbb{N} \rightarrow [0,1]$ and $s:\mathbb{N} \rightarrow \mathbb{N}$, with $s(n)=\Omega(\log n)$, all of which are computable in $\mathsf{DSPACE}(s)$. Let $\mathsf{Q_U SPACE}(s)_{c,k}$ denote the collection of all promise problems $\mathcal{P}=(\mathcal{P}_1,\mathcal{P}_0)$ such that there is a $\mathsf{DSPACE}(s)$-uniform family of (unitary) quantum circuits $\{Q_w:w \in \mathcal{P}\}$, where $Q_w$ acts on $h_w=O(s(\lvert w \rvert))$ qubits and has $2^{O(s(\lvert w \rvert))}$ gates, which has the following properties. The circuit $Q_w$ is applied to $h_w$ qubits that were initialized in the all-zeros state $\ket{0^{h_w}}$, after which the first qubit is measured in the standard basis. If the result is $1$, then we say $Q_w$ \textit{accepts} $w$; if, instead, the result is $0$, then we say $Q_w$ \textit{rejects} $w$. We require that the following conditions hold: 
		\begin{description}
			\item[Completeness:] $w \in \mathcal{P}_1 \Rightarrow \Pr[Q_w \text{ accepts } w] \geq c(\lvert w \rvert)$.
			\item[Soundness:] $w \in \mathcal{P}_0 \Rightarrow \Pr[Q_w \text{ accepts } w] \leq k(\lvert w \rvert)$.
		\end{description}	
		Let $\mathsf{BQ_U SPACE}(s)=\mathsf{Q_U SPACE}(s)_{\frac{2}{3},\frac{1}{3}}$ denote (two-sided) bounded-error \textit{unitary} quantum space $s$ and let $\mathsf{BQ_U L}=\mathsf{BQ_U SPACE}(\log n)$ denote unitary quantum \textit{logspace}.
	\end{definition}
	
	Note that, in the preceding definitions, $Q_w$ has $2^{O(s(\lvert w \rvert))}$ gates (this requirement is forced by the uniformity condition). That is to say, in our definition of quantum space $s(n)$, we require that the computation also runs in time $2^{O(s(\lvert w \rvert))}$; see the excellent survey paper by Saks \cite{saks1996randomization} for a thorough discussion of the desirability of requiring that space-bounded \textit{probabilistic} computations run in time at most exponential in their space bound, as well as, e.g., \cite{watrous1999space,watrous2003complexity,melkebeek2012time} for discussions of the analogous issue for \textit{quantum} computations. 	Note that the particular choice of universal gate set does not affect the definition of $\mathsf{BQ_U SPACE}(s)$ due to the space-efficient version \cite{melkebeek2012time} of the Solovay-Kitaev theorem. Furthermore, the class $\mathsf{BQ_U L}$ would remain the same if defined as $\mathsf{Q_U SPACE}(\log n)_{c,k}$ for any $c,k$ for which $c(n)=1-\frac{1}{poly(n)}$, $k(n)=\frac{1}{poly(n)}$, and $\exists q:\mathbb{N} \rightarrow \mathbb{N}_{>0}$, where $q(n)=poly(n)$, such that $c(n)-k(n)\geq \frac{1}{q(n)}$, $\forall n$ \cite{fefferman2016space}.  
	
	We next consider general space-bounded quantum computation. Most fundamentally, we wish to define a model of quantum computation that allows \textit{intermediate quantum measurements}. That is to say, rather than considering a purely unitary quantum computation in which a single measurement is performed at the end of the computation, we now allow measurements to be performed throughout the computation, and for the results of those measurements to be used to control the computation. As we wish for our main result (the fact that unitary and general space-bounded quantum computation are equivalent in power) to be as strong as possible, we wish to define a model of general space-bounded quantum computation that is as powerful as possible. To that end, we will define a space-bounded variant of the general quantum circuit model considered in the classic paper of Aharonov, Kitaev, and Nisan \cite{aharonov1998quantum}, in which gates are now arbitrary \textit{quantum channels}.
	
	\begin{definition}\label{def:generalQuantumCircuit}
		A \textit{general quantum circuit} on $h$ qubits is a sequence $\Phi=(\Phi_1,\ldots,\Phi_t)$ of quantum channels, where each $\Phi_j \in \chan(2^h)$. By slight abuse of notation, we use $\Phi$ to denote the element $\Phi_t \circ \cdots \circ \Phi_1 \in \chan(2^h)$ obtained by composing the individual gates of the circuit in order. We say that a family of general quantum circuits $\{\Phi_w=(\Phi_{w,1},\ldots,\Phi_{w,t_w}):w \in \mathcal{P}\}$ is $\mathsf{DSPACE}(s)$-uniform if there is a deterministic Turing machine that, on any input $w \in \mathcal{P}$, runs in space $O(s(\lvert w \rvert))$ (and hence time $2^{O(s(\lvert w \rvert))}$), and outputs a description of $\Phi_w$; to be precise, a description of $\Phi_w$ consists of the entries of each $K(\Phi_{w,j})$, where we require that $K(\Phi_{w,j}) \in \widehat{\mat}(2^{2h})$. 
	\end{definition}
	
	The operation of applying a unitary transformation is a special case of a quantum channel, and so the general quantum circuit model extends the unitary quantum circuit model. Moreover, the process of performing any (partial or full) quantum measurement is described by a quantum channel, and the preceding definition allows the results of intermediate measurements to be used to control which operations are applied at later stages of the computation (this can be accomplished by using a subset of the qubits as classical bits to store the results of earlier measurements, thereby making these results available to gates that appear later in the computation). It is necessary to establish some reasonable restriction on the complexity of computing a description of each gate of the circuit, as we do not wish to unreasonably increase the power of the model by allowing, for example, non-computable numbers to be used in defining each gate (see, e.g., \cite{watrous2003complexity,melkebeek2012time,remscrim2019lower,remscrim2019power} for further discussion of this issue). We note that this general quantum circuit model is equivalent to the space-bounded (general) quantum Turing machine model of Watrous \cite{watrous2003complexity}. We also note that the results of this paper would apply to any ``reasonable" variant of space-bounded quantum computation that is classically controlled, which includes all of the ``standard" variants that have been considered (e.g., \cite{watrous2003complexity,melkebeek2012time,ta2013inverting,perdrix2006classically}). We refer the reader to \cite[Section 2]{melkebeek2012time} for a thorough discussion of the various models of space-bounded quantum computation, and, in particular, of the reasonableness of requiring classical control.	 
	
	\begin{definition}\label{def:generalQuantumSpace}
		Consider functions $c,k:\mathbb{N} \rightarrow [0,1]$ and $s:\mathbb{N} \rightarrow \mathbb{N}$, with $s(n)=\Omega(\log n)$, all of which are computable in $\mathsf{DSPACE}(s)$. Let $\mathsf{Q SPACE}(s)_{c,k}$ denote the collection of all promise problems $\mathcal{P}=(\mathcal{P}_1,\mathcal{P}_0)$ such that there is a $\mathsf{DSPACE}(s)$-uniform family of general quantum circuits $\{\Phi_w:w \in \mathcal{P}\}$, where $\Phi_w$ acts on $h_w=O(s(\lvert w \rvert))$ qubits and has $2^{O(s(\lvert w \rvert))}$ gates, that has the following properties. The circuit $\Phi_w$ is applied to $h_w$ qubits that were initialized in the all-zeros state $\ket{0^{h_w}}$, after which the first qubit is measured in the standard basis. If the result is $1$, then $\Phi_w$ \textit{accepts} $w$; otherwise, $\Phi_w$ \textit{rejects} $w$. We require that $w \in \mathcal{P}_1 \Rightarrow \Pr[\Phi_w \text{ accepts } w] \geq c(\lvert w \rvert)$ and $w \in \mathcal{P}_0 \Rightarrow \Pr[\Phi_w \text{ accepts } w] \leq k(\lvert w \rvert)$. We define \textit{general} quantum space $\mathsf{BQSPACE}(s)=\mathsf{QSPACE}(s)_{\frac{2}{3},\frac{1}{3}}$ and $\mathsf{BQL}=\mathsf{BQSPACE}(\log n)$.
	\end{definition}
		
	We next define space-bounded quantum Merlin-Arthur proof systems, essentially following \cite{fefferman2018complete}. 
	
	\begin{definition}\label{def:spaceBoundedQMA}
		Consider functions $c,k:\mathbb{N} \rightarrow [0,1]$ and $s:\mathbb{N} \rightarrow \mathbb{N}$, with $s(n)=\Omega(\log n)$, all of which are computable in $\mathsf{DSPACE}(s)$. Let $\mathsf{Q_U MASPACE}(s)_{c,k}$ (resp. $\mathsf{QMASPACE}(s)_{c,k}$) denote the collection of all promise problems $\mathcal{P}=(\mathcal{P}_1,\mathcal{P}_0)$ such that there is a $\mathsf{DSPACE}(s)$-uniform family of unitary (resp. general) quantum circuits $\{V_w:w \in \mathcal{P}\}$, where $V_w$ acts on $m_w+h_w=O(s(\lvert w \rvert))$ qubits and has $2^{O(s(\lvert w \rvert))}$ gates, that has the following properties. Let $\Psi_{m_w}$ denote the set of $m_w$-qubit states. For each $w \in \mathcal{P}$, the verification circuit $V_w$ is applied to the state $\ket{\psi} \otimes \ket{0^{h_w}}$, where $\ket{\psi} \in \Psi_{m_w}$ is a (purported) proof of the fact that $w \in \mathcal{P}_1$. Then, the first qubit is measured in the standard basis. If the result is $1$, then $w$ is \textit{accepted}; otherwise, $w$ is \textit{rejected}. We require that $w \in \mathcal{P}_1 \Rightarrow \exists \ket{\psi} \in \Psi_{m_w}, \Pr[V_w \text{ accepts } w,\ket{\psi}] \geq c(\lvert w \rvert)$ and $w \in \mathcal{P}_0 \Rightarrow \forall \ket{\psi} \in \Psi_{m_w}, \Pr[V_w \text{ accepts } w,\ket{\psi}] \leq k(\lvert w \rvert)$.	We then define $\mathsf{Q_U MASPACE}(s)=\mathsf{Q_U MASPACE}(s)_{\frac{2}{3},\frac{1}{3}}$, $\mathsf{QMASPACE}(s)=\mathsf{QMASPACE}(s)_{\frac{2}{3},\frac{1}{3}}$, $\mathsf{Q_U MAL}=\mathsf{Q_U MASPACE}(\log n)$, and $\mathsf{QMAL}=\mathsf{QMASPACE}(\log n)$.
	\end{definition}
	
	Lastly we define analogues of the preceding classes for other error types. One-sided bounded-error: $\mathsf{RQ_U SPACE}(s)=\mathsf{Q_U SPACE}(s)_{\frac{1}{2},0}$, $\mathsf{RQMASPACE}(s)=\mathsf{QMASPACE}(s)_{\frac{1}{2},0}$, etc. One-sided unbounded-error: $\mathsf{NQSPACE}(s)=\bigcup\limits_{c:\mathbb{N} \rightarrow (0,1]} \mathsf{QSPACE}(s)_{c,0}$, $\mathsf{NQMASPACE}(s)=\bigcup\limits_{c:\mathbb{N} \rightarrow (0,1]} \mathsf{QMASPACE}(s)_{c,0}$, etc. Note that $\mathsf{RQMASPACE}$ and $\mathsf{NQMASPACE}$ have perfect soundness. We define $\mathsf{QMASPACE}$ with perfect completeness: $\mathsf{QMASPACE_1}(s)=\mathsf{QMASPACE}(s)_{1,\frac{1}{2}}$ and $\mathsf{PreciseQMASPACE_1}(s)=\bigcup\limits_{k:\mathbb{N} \rightarrow [0,1)} \mathsf{QMASPACE}(s)_{1,k}$. 
	
	\section{Well-Conditioned Determinant}\label{sec:determinant}
	
	We define the following well-conditioned versions of the standard $\mathsf{DET}^*$-complete problems \cite{cook1985taxonomy}. We first define well-conditioned versions of $\mathsf{DET}$ and $\mathsf{MATINV}$. The input to each problem consists of a matrix $A$ (among other values) which is promised to be well-conditioned (among other promises): $A \in \widehat{\mat}(n,1/\kappa(n),1)$. Recall that, by definition, this means that $\sigma_1(A) \leq 1$ and $\sigma_n(A) \geq 1/\kappa(n)$, which implies $A$ has condition number at most $\kappa(n)$.
	
	\begin{definition}\label{def:detPromiseProblems:detAndInv}
		Consider functions $\kappa:\mathbb{N} \rightarrow \mathbb{R}_{\geq 1}$ and $\epsilon:\mathbb{N} \rightarrow \mathbb{R}_{>0}$.
		\newline $\mathsf{DET}_{n,\kappa,\epsilon^{-1}}$
		\newline\hspace*{5pt}\textit{Input}: $A \in \widehat{\mat}(n)$, $b \in \mathbb{R}_{\leq 0}$
		\newline\hspace*{5pt}\textit{Promise}: $A \in \widehat{\mat}(n,1/\kappa(n),1)$, $\lvert \det(A) \rvert \in (0,e^{b-\epsilon(n)}] \cup [e^b,1]$
		\newline\hspace*{5pt}\textit{Output}: $1$ if $\lvert \det(A) \rvert \geq e^b$, $0$ otherwise	
		\newline $\mathsf{DET}_{n,\kappa,\epsilon^{-1}}^{+}$
		\newline\hspace*{5pt}\textit{Input}: $A \in \widehat{\mat}(n)$, $b \in \mathbb{R}_{\leq 0}$
		\newline\hspace*{5pt}\textit{Promise}: $A \in \widehat{\pos}(n,1/\kappa(n),1)$, $\det(A) \in (0,e^{b-\epsilon(n)}] \cup [e^b,1]$
		\newline\hspace*{5pt}\textit{Output}: $1$ if $\det(A) \geq e^b$, $0$ otherwise		 	
		\newline $\mathsf{MATINV}_{n,\kappa,\epsilon^{-1}}$
		\newline\hspace*{5pt}\textit{Input}: $A \in \widehat{\mat}(n)$, $s,t \in [n]$, $b \in \mathbb{R}_{\geq 0}$
		\newline\hspace*{5pt}\textit{Promise}: $A \in \widehat{\mat}(n,1/\kappa(n),1)$, $\lvert A^{-1}[s,t] \rvert \in [0,b-\epsilon(n)] \cup [b,\kappa(n)]$ 
		\newline\hspace*{5pt}\textit{Output}: $1$ if $\lvert A^{-1}[s,t]  \rvert \geq b$, $0$ otherwise	 	
		\newline $\mathsf{MATINV}_{n,\kappa,\epsilon^{-1}}^{+}$
		\newline\hspace*{5pt}\textit{Input}: $A \in \widehat{\mat}(n)$, $s,t \in [n]$, $b \in \mathbb{R}_{\geq 0}$
		\newline\hspace*{5pt}\textit{Promise}: $A \in \widehat{\pos}(n,1/\kappa(n),1)$, $A^{-1}[s,t] \in [0,b-\epsilon(n)] \cup [b,\kappa(n)]$ 
		\newline\hspace*{5pt}\textit{Output}: $1$ if $A^{-1}[s,t] \geq b$, $0$ otherwise
	\end{definition} 
	
	We next define well-conditioned versions of the various matrix multiplication problems; here, ``well-conditioned" has a somewhat different definition. For a sequence of matrices $A_1,\ldots,A_m$, and for indices $j_1,j_2$, where $1 \leq j_1 \leq j_2 \leq m$, let $A_{j_1,j_2}=\prod\limits_{j=j_1}^{j_2} A_j$. We require that all partial products have small singular values $\sigma_1(A_{j_1,j_2}) \leq \kappa(n)$.
	
	\begin{definition}\label{def:detPromiseProblems:powAndItprod}
		Consider functions $m:\mathbb{N} \rightarrow \mathbb{N}$, $\kappa:\mathbb{N} \rightarrow \mathbb{R}_{\geq 1}$, and $\epsilon:\mathbb{N} \rightarrow \mathbb{R}_{>0}$. 
		\newline $\mathsf{MATPOW}_{n,m,\kappa,\epsilon^{-1}}$
		\newline\hspace*{5pt}\textit{Input}: $A \in \widehat{\mat}(n)$, $s,t \in [n]$, $b \in \mathbb{R}_{\geq 0}$ 
		\newline\hspace*{5pt}\textit{Promise}: $\sigma_1(A^j) \leq \kappa(n) \, \forall j \in [m]$, $\lvert A^m[s,t] \rvert \in [0,b-\epsilon(n)] \cup [b,\kappa(n)]$
		\newline\hspace*{5pt}\textit{Output}: $1$ if $\lvert A^m[s,t] \rvert \geq b$, $0$ otherwise
		\newline $\mathsf{ITMATPROD}_{n,m,\kappa,\epsilon^{-1}}$
		\newline\hspace*{5pt}\textit{Input}: $A_1,\ldots,A_m \in \widehat{\mat}(n)$, $s,t \in [n]$, $b \in \mathbb{R}_{\geq 0}$ 
		\newline\hspace*{5pt}\textit{Promise}: $\sigma_1(A_{j_1,j_2}) \leq \kappa(n) \, \forall 1 \leq j_1 \leq j_2 \leq m$, $\left\lvert A_{1,m}[s,t] \right\rvert \in [0,b-\epsilon(n)] \cup [b,\kappa(n)]$
		\newline\textit{Output}: $1$ if $\left\lvert A_{1,m}[s,t] \right\rvert \geq b$, $0$ otherwise
		\newline $\mathsf{ITMATPROD}_{n,m,\kappa,\epsilon^{-1}}^{\geq 0}$
		\newline\hspace*{5pt}\textit{Input}: $A_1,\ldots,A_m \in \widehat{\mat}(n)$, $s,t \in [n]$, $b \in \mathbb{R}_{\geq 0}$ 
		\newline\hspace*{5pt}\textit{Promise}: $\sigma_1(A_{j_1,j_2}) \leq \kappa(n) \, \forall 1 \leq j_1 \leq j_2 \leq m$, $A_{1,m}[s,t] \in [0,b-\epsilon(n)] \cup [b,\kappa(n)]$ 
		\newline\hspace*{5pt}\textit{Output}: $1$ if $A_{1,m}[s,t] \geq b$, $0$ otherwise
		\newline $\mathsf{SUMITMATPROD}_{n,m,\kappa,\epsilon^{-1}}$
		\newline\hspace*{5pt}\textit{Input}: $A_1,\ldots,A_m \in \widehat{\mat}(n)$, $E \subseteq [n]^2$, $b \in \mathbb{R}_{\geq 0}$ 
		\newline\hspace*{5pt}\textit{Promise}: $\sigma_1(A_{j_1,j_2}) \leq \kappa(n) \, \forall 1 \leq j_1 \leq j_2 \leq m$, $\left\lvert \sum\limits_{(s,t) \in E} A_{1,m}[s,t] \right\rvert \in [0,b-\epsilon(n)] \cup [b,\lvert E \rvert \kappa(n)]$ 
		\newline\hspace*{5pt}\textit{Output}: $1$ if $\left\lvert \sum\limits_{(s,t) \in E} A_{1,m}[s,t] \right\rvert \geq b$, $0$ otherwise	
	\end{definition} 
	
	With the exception of the problem $\mathsf{DET}$ (and $\mathsf{DET}^{+}$), each of the above problems are defined such that they correspond to approximating some quantity with \textit{additive} error $\epsilon/2$; for example, $\mathsf{MATINV}$ involves determining if $\lvert A^{-1}[s,t]\rvert \leq b-\epsilon$ or  $\lvert A^{-1}[s,t]\rvert \geq b$. To clarify our definition of $\mathsf{DET}$, this problem corresponds to a $e^{\pm \frac{\epsilon}{2}}$ \textit{multiplicative} approximation of $\lvert \det(A) \rvert$, which is equivalent to an approximation of $\ln(\lvert \det(A) \rvert)$ with \textit{additive} error $\epsilon/2$. As we will see in \Cref{sec:determinant:bqulComplete} and \Cref{sec:determinant:approximateCounting}, this is the ``correct" definition of $\mathsf{DET}$, in the sense that it is the version of the determinant problem that most closely corresponds to the other linear-algebraic problems (matrix powering, matrix inversion, etc.). Moreover, these problems, defined as they are above, are somewhat ``over parameterized." For example, if $\langle A,s,t,b \rangle \in \mathsf{MATINV}_{n,\kappa(n),\epsilon^{-1}(n)}$, then $\mathsf{MATINV}(\langle A,s,t,b \rangle)=\mathsf{MATINV}(\langle \epsilon(n)A,s,t,\epsilon^{-1}(n)b \rangle)$ and $\langle \epsilon(n)A,s,t,\epsilon^{-1}(n)b \rangle \in \mathsf{MATINV}_{n,\kappa(n)\epsilon^{-1}(n),1}$. These additional parameters are convenient as they allow us to express certain results more cleanly. 
	
	\begin{definition}\label{def:polyConditionedDetPromiseProblems}
		For each promise problem $\mathcal{P}_{n,\kappa,\epsilon^{-1}}$ (resp. $\mathcal{P}_{n,m,\kappa,\epsilon^{-1}}$) in \Cref{def:detPromiseProblems:detAndInv,def:detPromiseProblems:powAndItprod}, we define $poly\text{-conditioned-}\mathcal{P}$ to be the promise problem $\mathcal{P}_{n,n^{O(1)},n^{O(1)}}$ (resp. $\mathcal{P}_{n,n^{O(1)},n^{O(1)},n^{O(1)}}$). For example, 
		\newline $poly\text{-conditioned-}\mathsf{DET}$
		\newline\hspace*{5pt}\textit{Input}: $A \in \widehat{\mat}(n)$ and $b \in \mathbb{R}_{\leq 0}$
		\newline\hspace*{5pt}\textit{Promise}: $A \in \widehat{\mat}(n,n^{-O(1)},1)$, $\lvert \det(A) \rvert \in (0,e^{b-n^{-O(1)}}] \cup [e^b,1]$
		\newline\hspace*{5pt}\textit{Output}: $1$ if $\lvert \det(A) \rvert \geq e^b$, $0$ otherwise
	\end{definition}
	
	We say that $\mathcal{P}=(\mathcal{P}_{n,f_1,\ldots,f_h})_{n \in \mathbb{N}}$ is (many-one) reducible to $\mathcal{P}'=(\mathcal{P}'_{m,f'_1,\ldots,f'_{h'}})_{m \in \mathbb{N}}$ if $\exists p_0,\ldots,p_{h'}$, where each $p_j$ is a real $(h+1)$-variate polynomial, such that $\forall n \in \mathbb{N}$, $\exists g_n:\mathcal{P}_{n,f_1,\ldots,f_h} \rightarrow \mathcal{P}'_{m,f'_1,\ldots,f'_{h'}}$ such that the following conditions hold: (1) $\mathcal{P}(\langle X \rangle)=\mathcal{P}'(g_n(\langle X \rangle))$, $\forall \langle X \rangle \in \mathcal{P}_{n,f_1,\ldots,f_h}$, (2) $m=p_0(n,f_1(n),\ldots,f_h(n))$, and (3) $f'_j(m)=p_j(n,f_1(n),\ldots,f_h(n))$, $\forall j$. If $(g_n)_{n \in \mathbb{N}}$ is computable in deterministic logspace (resp. uniform $\mathsf{NC}^1$, uniform $\mathsf{AC}^0$), we write $\mathcal{P} \leq_{\mathsf{L}}^m \mathcal{P}'$ (resp. $\mathcal{P} \leq_{\mathsf{NC}^1}^m \mathcal{P}'$, $\mathcal{P} \leq_{\mathsf{AC}^0}^m \mathcal{P}'$). For a complexity class $\mathsf{C}$, we say that $\mathcal{P}'$ is $\mathsf{C}$-complete if (1) $\mathcal{P}' \in \mathsf{C}$ and (2) $\mathcal{P} \leq_{\mathsf{L}}^m \mathcal{P}'$, $\forall \mathcal{P} \in \mathsf{C}$.
	
	Fefferman and Lin \cite{fefferman2018complete} showed that $poly\text{-conditioned-}\mathsf{MATINV}$ is $\mathsf{BQ_U L}$-complete. We extend their result by showing that by showing that all of the above $poly$-conditioned problems are $\mathsf{BQ_U L}$-complete. Along the way to proving this result, we also show that $\mathsf{BQ_U L}=\mathsf{BQL}$. To accomplish this, we will prove several lemmas that exhibit reductions between particular problems in \Cref{def:detPromiseProblems:detAndInv,def:detPromiseProblems:powAndItprod}. The proofs of these lemmas share a common structure: for a pair of promise problems $\mathcal{P},\mathcal{P}'$, we show how to transform an instance $w \in \mathcal{P}$ to an instance $f(w) \in \mathcal{P}'$ such that the reduction function $f$ preserves the answer (i.e., $\mathcal{P}(w)=\mathcal{P}'(f(w))$) and also preserves the property of being well-conditioned. Note that $\mathcal{P}  \leq_{\mathsf{L}}^m \mathcal{P}' \Rightarrow poly\text{-conditioned-}\mathcal{P} \leq_{\mathsf{L}}^m poly\text{-conditioned-}\mathcal{P}'$. 
	
	\subsection{Eliminating Intermediate Measurements}\label{sec:determinant:bqulBql}
	
	In this section, we show that intermediate measurements may be eliminated without any increase in needed space. We begin by showing that $poly\text{-conditioned-}\mathsf{ITMATPROD}$ is $\mathsf{BQL}$-hard and is in $\mathsf{BQ_U L}$, which implies $\mathsf{BQL}=\mathsf{BQ_U L}$; the general equivalence then follows from a standard padding argument. In the following, we assume that $m(n)$, $\kappa(n)$, and $\epsilon(n)^{-1}$ can be computed to $O(\log n)$ bits of precision in uniform $\mathsf{AC}^0$. For $m \in \mathbb{N}_{\geq 1}$ and $r,c \in [m]$, define $F_{m,r,c} \in \widehat{\mat}(m)$ such that $F_{m,r,c}[r,c]=1$ and $F_{m,r,c}[r',c']=0$, $\forall (r',c') \neq (r,c)$.
	
	\begin{lemma}\label{lemma:reduction:iterMatProdToMatPow}
		$\mathsf{ITMATPROD} \leq_{\mathsf{AC}^0}^m \mathsf{MATPOW}$.
	\end{lemma}
	\begin{proof}
		Consider $\langle A_1,\ldots,A_m,s,t,b \rangle \in \mathsf{ITMATPROD}_{n,m,\kappa,\epsilon^{-1}}$. Following \cite{cook1985taxonomy}, let $\widehat{A}=\sum\limits_{r=1}^m F_{m+1,r,r+1} \otimes A_r \in \widehat{\mat}(nm+n)$ consist of $n \times n$ blocks, where the blocks immediately above the main diagonal blocks are given by $A_1,\ldots,A_m$, and all other entries are $0$. For $j \in [m]$, we have $$\widehat{A}^j=\sum_{r=1}^{m+1-j} F_{m+1,r,r+j} \otimes A_{r,r+j-1}.$$ 
		Let $\widehat{s}=s$, $\widehat{t}=nm+t$, $\widehat{b}=b$. Then $\widehat{A}^m[\widehat{s},\widehat{t}]=A_{1,m}[s,t]$, which implies $\mathsf{ITMATPROD}(\langle A_1,\ldots,A_m,s,t,b \rangle) =\mathsf{MATPOW}(\langle \widehat{A},\widehat{s},\widehat{t},\widehat{b} \rangle)$. 	
		Moreover, $\forall j \in [m]$, we have $$\sigma_1(\widehat{A}^j)=\sigma_1\left(\sum_{r=1}^{m+1-j} F_{m+1,r,r+j} \otimes A_{r,r+j-1}\right)=\sigma_1\left(\bigoplus_{r=1}^{m+1-j} A_{r,r+j-1}\right)=\max_{r} \sigma_1(A_{r,r+j-1}) \leq \kappa(n).$$		
		Therefore, $\langle \widehat{A},\widehat{s},\widehat{t},\widehat{b} \rangle \in \mathsf{MATPOW}_{nm+n,m,\kappa(n),\epsilon^{-1}(n)}$.
	\end{proof}
	
	\begin{lemma}\label{lemma:reduction:matPowToMatInv}
		$\mathsf{MATPOW}\leq_{\mathsf{AC}^0}^m \mathsf{MATINV}$.
	\end{lemma}
	\begin{proof}
		Consider $\langle A,s,t,b \rangle \in \mathsf{MATPOW}_{n,m,\kappa,\epsilon^{-1}}$. Following \cite{cook1985taxonomy}, let $G_j=\sum\limits_{r=1}^{m+1-j} F_{m+1,r,r+j} \in \widehat{\mat}(m+1)$, $\forall j \in [m]$. Let $Y=G_1 \otimes A \in \widehat{\mat}(nm+n)$ consist of $n \times n$ blocks, where the blocks immediately above the main diagonal blocks are all given by $A$. Let $Z=I_{nm+n}-Y \in \widehat{\mat}(nm+n)$ and observe that $$Z^{-1}=\sum_{j=0}^m G_j \otimes A^j.$$ Let $\widehat{s}=s$ and $\widehat{t}=nm+t$. Then $Z^{-1}[\widehat{s},\widehat{t}]=A^m[s,t]$. We have $\sigma_1(Z)\leq \sigma_1(I_{nm+n})+\sigma_1(Y)\leq 1+\kappa(n)$ and $$\sigma_1(Z^{-1})=\sigma_1\left(\sum_{j=0}^m G_j \otimes A^j\right)\leq \sum_{j=0}^m \sigma_1(G_j \otimes A^j) \leq 1+\sum_{j=1}^m \sigma_1(A^j)\leq 1+\sum_{j=1}^m \kappa(n)\leq 1+m\kappa(n).$$ This implies $\sigma_{nm+n}(Z)=\sigma_1(Z^{-1})^{-1} \geq (1+m\kappa(n))^{-1}$. Let $\widehat{Z}=\frac{1}{\lceil 1+\kappa(n) \rceil} Z  \in \widehat{\mat}(nm+n)$ and $\widehat{b}=\lceil 1+\kappa(n) \rceil b$. We then conclude that $\langle \widehat{Z},\widehat{s},\widehat{t},\widehat{b} \rangle \in \mathsf{MATINV}_{nm+n,(1+m\kappa(n))\lceil 1+\kappa(n)\rceil,\lceil 1+\kappa(n)\rceil^{-1} \epsilon^{-1}(n)}$ and $\mathsf{MATPOW}(\langle A,s,t,b \rangle) = \mathsf{MATINV}(\langle \widehat{Z},\widehat{s},\widehat{t},\widehat{b} \rangle)$.
	\end{proof}
	
	\begin{lemma}\label{lemma:reduction:matInvToPosMatInv}
		$\mathsf{MATINV} \leq_{\mathsf{NC}^1}^m \mathsf{MATINV}^{+}$.
	\end{lemma}
	\begin{proof}
		Consider $\langle A,s,t,b \rangle \in \mathsf{MATINV}_{\kappa,\epsilon^{-1}}$. We define $\widehat{H} =\frac{1}{3}\begin{pmatrix}
		A^{\dagger} A & - A^{\dagger} \\
		-A & 2I
		\end{pmatrix} \in \widehat{\pos}(2n)$. Note that $\widehat{H}^{-1} =3\begin{pmatrix}
		2(A^{\dagger} A)^{-1} & A^{-1} \\
		(A^{\dagger})^{-1} &  I
		\end{pmatrix}$. Moreover, $\sigma_1(\widehat{H}) \leq 1$ and $\sigma_{2n}(\widehat{H}) \geq \frac{1}{9} (\sigma_n(A))^2 \geq (3 \kappa(n))^{-2}$. Let $\widehat{s}=s$, $\widehat{t}=t+n$, and $\widehat{b}=3b$. Then $\widehat{H}^{-1}[\widehat{s},\widehat{t}]=3A^{-1}[s,t]$. Therefore, $\mathsf{MATINV}(\langle A,s,t,b \rangle) =  \mathsf{MATINV}(\langle H,\widehat{s},\widehat{t},\widehat{b} \rangle)$ and $\langle \widehat{H},\widehat{s},\widehat{t},\widehat{b} \rangle \in \mathsf{MATINV}_{2n,(3\kappa(n))^2,(3\epsilon(n))^{-1}}^{+}$.
	\end{proof}
	
	\begin{lemma}\label{lemma:itMatProdInBqul}
		$poly\text{-conditioned-}\mathsf{ITMATPROD} \in \mathsf{BQ_U L}$. 
	\end{lemma}
	\begin{proof}
		By \Cref{lemma:reduction:iterMatProdToMatPow,lemma:reduction:matPowToMatInv,lemma:reduction:matInvToPosMatInv}, $\mathsf{ITMATPROD} \leq_{\mathsf{NC}^1}^m \mathsf{MATINV}^{+}$. Recall $\mathcal{P}  \leq_{\mathsf{NC}^1}^m \mathcal{P}' \Rightarrow poly\text{-conditioned-}\mathcal{P} \leq_{\mathsf{NC}^1}^m poly\text{-conditioned-}\mathcal{P}'$. By \cite[Theorem 13]{fefferman2018complete}, $poly\text{-conditioned-}\mathsf{MATINV}^{+} \in \mathsf{BQ_U L}$.
	\end{proof}
	
	\begin{lemma}\label{lemma:itMatProdIsBqlHard}
		$poly\text{-conditioned-}\mathsf{ITMATPROD}$ is $\mathsf{BQL}$-hard. 
	\end{lemma}
	\begin{proof}
		Suppose $\mathcal{P}=(\mathcal{P}_1,\mathcal{P}_0) \in \mathsf{BQL}$. By definition, there is a uniform family of general quantum circuits $\{\Phi_w=(\Phi_{w,1},\ldots,\Phi_{w,t_w}):w \in \mathcal{P}\}$, where $\Phi_w$ acts on $h_w=O(\log \lvert w \rvert)$ qubits and has $t_w=\lvert w \rvert^{O(1)}$ gates, such that if $w \in \mathcal{P}_1$, then $\Pr[\Phi_w \text{ accepts } w]\geq \frac{2}{3}$, and if $w \in \mathcal{P}_0$, then $\Pr[\Phi_w \text{ accepts } w]\leq \frac{1}{3}$. Without loss of generality we may, for convenience, assume that $\Phi_w$ ``cleans-up" its workspace at the end of the computation, by measuring the first qubit in the computational basis, and then forcing every other qubit to the state $\ket{0}$ (by measuring each such qubit in the computational basis and, if the result $1$ is obtained, flipping its value). 
		
		Let $d_w=2^{2h_w}=\lvert w \rvert^{O(1)}$. For each $j \in [t_w]$, we define $A(w)_{j}=K(\Phi_{w,t_w-j+1})$; note that, by \Cref{def:generalQuantumCircuit}, $A(w)_{j} \in \widehat{\mat}(d_w)$ and $A(w)_{j}$ can be constructed in $\mathsf{DSPACE}(\log (\lvert w \rvert))$. Moreover, as $\Phi_{w,j}\in \chan(2^{h_w})$, $\forall j \in [t_w]$, we have $\Phi_{w,t_w-j_2+1} \circ \cdots \circ \Phi_{w,t_w-j_1+1}\in\chan(2^{h_w})$, which by \cite[Theorem 1]{roga2013entropic} implies the following bound on the largest singular value of any partial product of the $A(w)_j$, $$\sigma_1(A(w)_{j_1,j_2})=\sigma_1\left(\prod\limits_{j=j_1}^{j_2}  A(w)_{j}\right)=\sigma_1(K(\Phi_{w,t_w-j_2+1} \circ \cdots \circ \Phi_{w,t_w-j_1+1}))\leq \sqrt{d_w}=n^{O(1)}.$$ 
		Let $x_w=\ket{10^{h_w-1}}\bra{10^{h_w-1}}$ and $y_w=\ket{0^{h_w}}\bra{0^{h_w}}$. By \Cref{def:generalQuantumSpace}, $$\Pr[\Phi_w \text{ accepts } w]=\left(\prod\limits_{j=1}^{t_w} A(w)_{j}\right)[x_w,y_w]=A(w)_{1,t_w}[x_w,y_w].$$ 
		Therefore, $\mathsf{ITMATPROD}(\langle A(w)_{1},\ldots,A(w)_{t_w},x_w,y_w,\frac{2}{3} \rangle)= \mathcal{P}(w)$ and $\langle A(w)_{1},\ldots,A(w)_{t_w},x_w,y_w,\frac{2}{3} \rangle \in poly\text{-conditioned-}\mathsf{ITMATPROD}$.
	\end{proof}
	
	For the sake of completeness, in \Cref{sec:appendix:turingMachineVersionOfBqlEqualsBqul}, we also prove a version of the preceding lemma for the (equivalent) version of quantum logspace that is defined using quantum Turing machines.
	
	\begin{lemma}\label{lemma:qmalInQumalBqul}
		$\mathsf{QMAL} \subseteq \mathsf{Q_U MAL}^{\mathsf{BQ_U L}}$.
	\end{lemma}
	\begin{proof}
		This follows by an argument similar to that of the proof of \Cref{lemma:itMatProdIsBqlHard}. Suppose $\mathcal{P}=(\mathcal{P}_1,\mathcal{P}_0) \in \mathsf{QMAL}$. There is a uniform family of general quantum circuits $\{\Phi_w=(\Phi_{w,1},\ldots,\Phi_{w,t_w}):w \in \mathcal{P}\}$, where $\Phi_w$ acts on $m_w+h_w=O(s(\lvert w \rvert))$ qubits and has $2^{O(s(\lvert w \rvert))}$ gates, that has the following properties. Let $\Pi_1=\ket{1}\bra{1} \otimes I_{2^{m_w+h_w-1}}$ and let $\Psi_{m_w}$ denote the set of $m_w$-qubit states. For each $w \in \mathcal{P}$, the verification circuit $\Phi_w$ is applied to the state $\ket{\psi} \otimes \ket{0^{h_w}}$, where $\ket{\psi} \in \Psi_{m_w}$ is a (purported) proof of the fact that $w \in \mathcal{P}_1$. Then, the first qubit is measured in the standard basis. If the result is $1$, then $w$ is \textit{accepted}; otherwise, $w$ is \textit{rejected}. If $w \in \mathcal{P}_1$, then $\exists \ket{\psi} \in \Psi_{m_w}, \Pr[\Phi_w \text{ accepts } w,\ket{\psi}] \geq \frac{2}{3}$, and if $w \in \mathcal{P}_0$, then $\forall \ket{\psi} \in \Psi_{m_w}, \Pr[\Phi_w \text{ accepts } w,\ket{\psi}] \leq \frac{1}{3}$.	
		
		We have $\Pr[\Phi_w \text{ accepts } w,\ket{\psi}]= \langle \Pi_1, \Phi_w (\ket{\psi}\bra{\psi} \otimes \ket{0^{h_w}}\bra{0^{h_w}}) \rangle$ $= \langle \Phi_w^{\dagger}(\Pi_1), \ket{\psi}\bra{\psi} \otimes \ket{0^{h_w}}\bra{0^{h_w}} \rangle$. Let $M=(I \otimes \bra{0^{h_w}}) \Phi_w^{\dagger}(\Pi_1) (I \otimes \ket{0^{h_w}}) \in \widehat{\pos}(2^{m_w})$. Then $w \in \mathcal{P}_1 \Leftrightarrow \lambda_1(M) \geq \frac{2}{3}$ and $w \in \mathcal{P}_0 \Leftrightarrow \lambda_1(M) \leq \frac{1}{3}$. Given access to $M$, the problem, of determining if $\lambda_1(M) \geq \frac{2}{3}$ or $\lambda_1(M) \leq \frac{1}{3}$, is obviously in $\mathsf{Q_U MAL}$. By the same argument as in the proof of \Cref{lemma:itMatProdIsBqlHard}, estimating an entry of $M$ (to $1/poly(n)$ precision) corresponds to an instance of $poly\text{-conditioned-}\mathsf{ITMATPROD}$, which, by \Cref{lemma:itMatProdInBqul}, is in $\mathsf{BQ_U L}$. Therefore, $\mathcal{P} \in \mathsf{Q_U MAL}^{\mathsf{BQ_U L}}$.	 
	\end{proof}
	
	\begin{lemma}\label{thm:bqlEqualsBqulEqualsQmal}
		$\mathsf{BQ_U L} = \mathsf{BQL} =\mathsf{Q_U MAL} = \mathsf{QMAL}$.
	\end{lemma}
	\begin{proof}
		Clearly, $\mathsf{BQ_U L} \subseteq \mathsf{BQL}$. By \Cref{lemma:itMatProdIsBqlHard}, $poly\text{-conditioned-}\mathsf{ITMATPROD}$ is $\mathsf{BQL}$-hard, and, by \Cref{lemma:itMatProdInBqul}, $poly\text{-conditioned-}\mathsf{ITMATPROD} \in \mathsf{BQ_U L}$; this implies $\mathsf{BQL} \subseteq \mathsf{BQ_U L}$. Therefore, $\mathsf{BQL} = \mathsf{BQ_U L}$. By \cite[Theorem 18]{fefferman2018complete} $\mathsf{Q_U MAL} = \mathsf{BQ_U L}$ (in fact, their argument shows $\mathsf{Q_U MAL}^{\mathcal{O}} = \mathsf{BQ_U L}^{\mathcal{O}}$, for any oracle $\mathcal{O}$). Clearly, $\mathsf{BQL} \subseteq \mathsf{QMAL}$, and so it suffices to show $\mathsf{QMAL} \subseteq \mathsf{BQL}$. By \Cref{lemma:qmalInQumalBqul}, and the (straightforward) fact that $\mathsf{BQL}$ is self-low, we have 
		\[\mathsf{QMAL} \subseteq \mathsf{Q_U MAL}^{\mathsf{BQ_U L}}=\mathsf{BQ_U L}^{\mathsf{BQ_U L}}\subseteq \mathsf{BQL}^{\mathsf{BQL}}=\mathsf{BQL}.\qedhere\]
	\end{proof}
	
	We now prove \Cref{thm:intro:bquspaceEqualsBqspace} from \Cref{sec:intro:elimMeasurements}, which we restate here for convenience.
	\restateBquspaceEqualsBqspace*
	\begin{proof}
		Clearly, $\mathsf{BQ_U SPACE}(s) \subseteq \mathsf{BQSPACE}(s)$. The containment $\mathsf{BQSPACE}(s) \subseteq \mathsf{BQ_U SPACE}(s)$ follows from \Cref{thm:bqlEqualsBqulEqualsQmal} and a standard padding argument; for the sake of completeness, we briefly state this argument in \Cref{sec:appendix:proofOfTheorem1}. Analogous statements hold for $\mathsf{QMASPACE}$.
	\end{proof}
	
	\subsection{B\texorpdfstring{Q\textsubscript{U}}{QU}L Completeness}\label{sec:determinant:bqulComplete}
	
	We next show that the $poly$-conditioned versions of all of the standard $\mathsf{DET}^*$-complete problems (\Cref{def:detPromiseProblems:detAndInv,def:detPromiseProblems:powAndItprod}) are $\mathsf{BQ_U L}$-complete. By \cite[Theorem 13]{fefferman2018complete}, $poly\text{-conditioned-}\mathsf{MATINV}$ is $\mathsf{BQ_U L}$-complete\footnote{Note that Ref. \cite{fefferman2018complete} showed $poly\text{-conditioned-}\mathsf{MATINV}^{+} \in \mathsf{BQ_U L}$ and that $poly\text{-conditioned-}\mathsf{MATINV}$ is $\mathsf{BQ_U L}$-hard, but the equivalence between these two problems is ``obvious" (and shown explicitly in \Cref{lemma:reduction:matInvToPosMatInv})}; therefore, it suffices to exhibit a chain of reductions through the various $poly\text{-conditioned-}\mathcal{P}$. Recall that $\mathcal{P}  \leq_{\mathsf{L}}^m \mathcal{P}' \Rightarrow poly\text{-conditioned-}\mathcal{P} \leq_{\mathsf{L}}^m poly\text{-conditioned-}\mathcal{P}'$
	
	\begin{lemma}\label{lemma:reduction:posDetToSumIterMatProd}
		$\mathsf{DET}^{+} \leq_{\mathsf{AC}^0}^m \mathsf{SUMITMATPROD}$.
	\end{lemma}
	\begin{proof}	
		Consider $\langle H,b \rangle \in \mathsf{DET}_{n,\kappa,\epsilon^{-1}}^{+}$. By the promise, $H \in \widehat{\pos}(n)$, $\lambda_1(H)=\sigma_1(H)\leq 1$, and $\lambda_n(H)=\sigma_n(H) \geq \kappa(n)^{-1}$, which implies  $\sigma_1(I-H)=\lambda_1(I-H)=1-\lambda_n(H) \leq 1-\kappa(n)^{-1}<1$. This implies $\ln(H)=-\sum\limits_{k=1}^{\infty} \frac{(I-H)^k}{k}$, where here $\ln(H)$ denotes the matrix logarithm. Recall that, as a consequence of Jacobi's formula, $\ln(\det(H))=\tr(\ln(H))$. 
		
		For $m \in \mathbb{N}_{\geq 1}$, let $S_m=\sum\limits_{k=1}^{m} \frac{(I-H)^k}{k}$, let $R_m=\sum\limits_{k=m+1}^{\infty} \frac{(I-H)^k}{k}=-\log(H)-S_m$, and let $D_m \in \widehat{\mat}(m)$ denote the diagonal matrix where $D_m[k,k]=\frac{1}{k}$. Let $\widehat{l}=\lfloor 1+\log(\lfloor \kappa(n) \rfloor) \rfloor$, let $\widehat{A}_1=I_{n\widehat{l}} \oplus (-D_m \otimes (I-H)) \in \widehat{\mat}(n\widehat{l}+nm)$, and, for $k \in [m]$, let $\widehat{A}_k=I_{n(\widehat{l}+k-1)} \oplus (I_{m+1-k} \otimes (I-H))\in \widehat{\mat}(n\widehat{l}+nm)$. Then $$\widehat{A}_{1,m}=\prod_{j=1}^m \widehat{A}_j=I_{n\widehat{l}} \oplus \left(\bigoplus_{k=1}^m \frac{-(I-H)^k}{k}\right).$$ 
		Let $E_m=\{(d,d):d \in [n\widehat{l}+nm]\}$. We then have $$\sum_{(s,t) \in E_m} \widehat{A}_{1,m}[s,t]=\tr(\widehat{A}_{1,m})=\tr(I_{n\widehat{l}})-\sum_{k=1}^m \tr\left(\frac{(I-H)^k}{k}\right)=n\widehat{l}-\tr(S_m)=n\widehat{l}+\ln(\det(H))+\tr(R_m).$$
		Moreover, for $1 \leq j_1 \leq j_2 \leq m$, we have $$\sigma_1(\widehat{A}_{j_1,j_2}) \leq \max\left(\sigma_1(I_{n\widehat{l}}),\max_{k \in \{0,\ldots,j_2-j_1\}} \sigma_1((I-H)^k)\right) = 1.$$ 		
		As shown above, $\sigma_1(I-H) \leq 1-\kappa(n)^{-1}$, which implies
		$$\sigma_1(R_m)= \sigma_1\left(\sum_{k=m+1}^{\infty} \frac{(I-H)^k}{k} \right) \leq \sum_{k=m+1}^{\infty}  \frac{(\sigma_1(I-H))^k}{k} \leq \sum_{k=m+1}^{\infty} \frac{(1-\kappa(n)^{-1})^k}{k}\leq \kappa(n) \left(1-\frac{1}{\kappa(n)}\right)^{m+1}.$$
		If $m \geq \kappa(n) \ln(2n\kappa(n)\epsilon(n)^{-1})$, then $$\tr(R_m)\leq n \sigma_1(R_m) \leq n \kappa(n) \left(1-\frac{1}{\kappa(n)}\right)^{\kappa(n) \ln(2n\kappa(n)\epsilon(n)^{-1})} \leq n\kappa(n) \left(\frac{1}{e}\right)^{\ln(2n\kappa(n)\epsilon(n)^{-1})} =\frac{1}{2} \epsilon(n).$$ 
		Let $\widehat{m}=\lceil \kappa(n) \rceil \lfloor 1+\log(\lfloor 2n\kappa(n)\epsilon(n)^{-1} \rfloor) \rfloor \geq \kappa(n) \ln(2n\kappa(n)\epsilon(n)^{-1})$ and $\widehat{E}=E_{\widehat{m}}$. Note that $\tr(R_m) \geq 0$. We then have, $$n\widehat{l}+\ln(\det(H)) \leq \sum_{(s,t) \in \widehat{E}} \widehat{A}_{1,m}[s,t]= n\widehat{l}+\ln(\det(H))+\tr(R_{\widehat{m}}) \leq n\widehat{l}+\ln(\det(H))+\frac{1}{2} \epsilon(n).$$ 
		If $\det(H) \geq e^b$, then $\sum\limits_{(s,t) \in \widehat{E}} \widehat{A}_{1,m}[s,t] \geq n\widehat{l}+b$; if $\det(H) \leq e^{b-\epsilon(n)}$, then $\sum\limits_{(s,t) \in \widehat{E}} \widehat{A}_{1,m}[s,t] \leq n\widehat{l}+b-\frac{1}{2}\epsilon(n)$. Let  $\widehat{b}=n\widehat{l}+b$. Therefore, $\langle \widehat{A}_1,\ldots,\widehat{A}_{\widehat{m}}, \widehat{E},\widehat{b} \rangle \in \mathsf{SUMITMATPROD}_{n(\widehat{l}+\widehat{m}),\widehat{m},1,2 \epsilon^{-1}(n)}$ and $\mathsf{DET}(\langle H,b \rangle) = \mathsf{SUMITMATPROD}(\langle \widehat{A}_1,\ldots,\widehat{A}_{\widehat{m}}, \widehat{E},\widehat{b} \rangle)$. 
	\end{proof}	
	
	\begin{lemma}\label{lemma:reduction:iterMatProdToNonnegIterMatProd}
		$\mathsf{ITMATPROD} \leq_{\mathsf{AC}^0}^m \mathsf{ITMATPROD}^{\geq 0}$.
	\end{lemma}
	\begin{proof}
		Consider $\langle A_1,\ldots,A_m,s,t,b \rangle \in \mathsf{ITMATPROD}_{n,m,\kappa,\epsilon^{-1}}$. For $j \in [2m+1]$, we define $$\widehat{A}_j=\begin{cases}
		A_j, & j\leq m \\
		\ket{t}\bra{t}, & j=m+1\\
		A_{2m+2-j}^{\dagger}, & j\geq m+2
		\end{cases}$$ We then have $\widehat{A}_{1,2m+1}[s,s]=A_{1,m}[s,t]\overline{A_{1,m}[s,t]}=\lvert A_{1,m}[s,t] \rvert^2$. Consider $j_1,j_2$ such that $1 \leq j_1 \leq j_2 \leq m+1$. If $j_1 \leq m+1 \leq j_2$, then $$\sigma_1(\widehat{A}_{j_1,j_2})=\sigma_1\left(A_{j_1,m}\ket{t}\bra{t}A_{j_2,m}^{\dagger}\right)\leq \sigma_1(A_{j_1,m})\sigma_1(\ket{t}\bra{t})\sigma_1\left(A_{j_2,m}^{\dagger}\right) \leq \kappa(n)^2.$$ If $j_2<m+1$, then $\sigma_1(\widehat{A}_{j_1,j_2})=\sigma_1(A_{j_1,j_2}) \leq \kappa(n)$. If $j_1>m+1$, then $\sigma_1(\widehat{A}_{j_1,j_2})=\sigma_1(A_{j_1-m-1,j_2-m-1}) \leq \kappa(n)$. Let $\widehat{s}=\widehat{t}=s$ and $\widehat{b}=b^2$. Then $\langle \widehat{A}_1,\ldots,\widehat{A}_{2m+1},\widehat{s},\widehat{t},\widehat{b} \rangle \in \mathsf{ITMATPROD}_{n,2m+1,\kappa^2,\epsilon^{-2}}^{\geq 0}$ and $\mathsf{ITMATPROD}(\langle A_1,\ldots,A_m,s,t,b \rangle)=  \mathsf{ITMATPROD}(\langle \widehat{A}_1,\ldots,\widehat{A}_{2m+1},\widehat{s},\widehat{t},\widehat{b} \rangle)$.
	\end{proof}
	
	\begin{lemma}\label{lemma:reduction:nonnegIterMatProdToDet}
		$\mathsf{ITMATPROD}^{\geq 0} \leq_{\mathsf{AC}^0}^m \mathsf{DET}$.
	\end{lemma}
	\begin{proof}	
		Consider $\langle A_1,\ldots,A_m,s,t,b \rangle \in \mathsf{ITMATPROD}_{n,m,\kappa,\epsilon^{-1}}^{\geq 0}$. Let $Y=\sum\limits_{r=1}^m F_{m+1,r,r+1} \otimes A_r \in \widehat{\mat}(nm+n)$ consist of $n \times n$ blocks, where the blocks immediately above the main diagonal blocks are given by $A_1,\ldots,A_m$, and all other entries are $0$. Let $B=I_{nm+n}-Y \in \widehat{\mat}(nm+n)$ and observe that $$B^{-1}=I_{nm+n}+\sum_{r=1}^{m} \sum_{c=r+1}^{m+1} F_{m+1,r,c}\otimes A_{r,c-1}.$$ 
		Let $C=B+\ket{nm+t}\bra{s}$. By the matrix determinant lemma, and the fact that $\det(B)=1$, $$\det(C)=(1+\bra{s} B^{-1} \ket{nm+t}) \det(B)=1+B^{-1}[s,nm+t]=1+A_{1,m}[s,t].$$
		Next, observe that $$\sigma_1(C) \leq \sigma_1(\ket{nm+t}\bra{s})+\sigma_1(I)+\sigma_1(Y)\leq 2+\max_j\sigma_1(A_j)\leq 2+\kappa(n).$$ 
		Notice that $B^{-1}=\sum\limits_{j=0}^m Y^j$, which implies $$\sigma_1(B^{-1})\leq \sum_{j=0}^m \sigma_1(Y^j)\leq 1+\sum_{j=1}^m \left( \max_{k \in [m-j+1]} \sigma_1(A_{k,k+j-1}) \right) \leq 1+\sum_{j=1}^m \kappa(n) =1+m\kappa(n).$$ 
		By the Sherman-Morrison formula, $C^{-1}=B^{-1}(I-(1+\bra{s} B^{-1}\ket{nm+t})^{-1} \ket{nm+t}\bra{s} B^{-1})$. Recall that, by the promise, $\bra{s} B^{-1}\ket{nm+t}=A_{1,m}[s,t] \in \mathbb{R}_{\geq 0}$. Therefore, $$\sigma_1(C^{-1}) \leq \sigma_1(B^{-1})(\sigma_1(I)+\sigma_1((1+\bra{s} B^{-1}\ket{nm+t})^{-1} \ket{nm+t}\bra{s})\sigma_1(B^{-1})) \leq (1+m\kappa(n))(2+m\kappa(n)).$$ This implies $\sigma_{nm+n}(C)=\sigma_1(C^{-1})^{-1} \geq ((1+m\kappa(n))(2+m\kappa(n)))^{-1}$. Let $\widehat{l}=\lfloor 1+\ln(\lfloor 2+\kappa(n) \rfloor) \rfloor$ and let $\widehat{C}=e^{-\widehat{l}}C \in \widehat{\mat}(nm+n)$. Then, for $j \in [nm+n]$, $\sigma_j(\widehat{C})=e^{-\widehat{l}}\sigma_j(C)$; in particular, $\sigma_1(\widehat{C}) \leq 1$ and $\sigma_{nm+n}(\widehat{C}) \geq (2+m\kappa(n))^{-3}$. Moreover, $$\lvert \det(\widehat{C})\rvert =\lvert e^{-\widehat{l}(nm+n)} \det(C) \rvert=\lvert e^{-\widehat{l}(nm+n)} (1+A_{1,m}[s,t]) \rvert=e^{-\widehat{l}(nm+n)}(1+A_{1,m}[s,t]).$$	
		Let $\widehat{a}=\ln(1+b-\epsilon(n))-\widehat{l}(nm+n)$ and $\widehat{b}=\ln(1+b)-\widehat{l}(nm+n)$. If $A_{1,m}[s,t] \geq b$, then $\lvert \det(\widehat{C}) \rvert \geq e^{\widehat{b}}$; if $A_{1,m}[s,t] \leq b-\epsilon(n)$, then $\lvert \det(\widehat{C}) \rvert \leq e^{\widehat{a}}$; thus, $\mathsf{ITMATPROD}(\langle A_1,\ldots,A_m,s,t,b \rangle) =  \mathsf{DET}(\langle \widehat{C},\widehat{b} \rangle)$. Finally, $$\widehat{b}-\widehat{a}=\ln\left(\frac{1+b}{1+ b-\epsilon(n)}\right)=\ln\left(1+\frac{\epsilon(n)}{1+ b-\epsilon(n)}\right)\geq \ln\left(1+\frac{\epsilon(n)}{1+\kappa(n)}\right)\geq \frac{\epsilon(n)}{2(1+\kappa(n))}.$$  Therefore, $\langle \widehat{C},\widehat{b} \rangle \in \mathsf{DET}_{nm+m,(2+m\kappa(n))^3,\epsilon^{-1}(n)(2+2\kappa(n))}$. 
	\end{proof}
	
	\begin{lemma}\label{lemma:reduction:detToPosDet}
		$\mathsf{DET} \leq_{\mathsf{NC}^1}^m \mathsf{DET}^{+}$.
	\end{lemma}
	\begin{proof}
		Consider $\langle A,b \rangle \in \mathsf{DET}_{n,\kappa,\epsilon^{-1}}$. Let $\widehat{H}=AA^{\dagger} \in \widehat{\pos}(n)$ and $\widehat{b}=2b$. Then, $\det(\widehat{H})=\lvert \det(A) \rvert^2$ and $\sigma_j(\widehat{H})=\sigma_j^2(A), \ \forall j$. Therefore, $\langle \widehat{H},\widehat{b} \rangle \in \mathsf{DET}_{n,\kappa^2,2\epsilon^{-1}}$ and $\mathsf{DET}(\langle A,b \rangle) = \mathsf{DET}(\langle \widehat{H},\widehat{b} \rangle)$.
	\end{proof}
	
	\begin{lemma}\label{lemma:reduction:posMatInvToSumIterMatProd}
		$\mathsf{MATINV}^{+} \leq_{\mathsf{AC}^0}^m \mathsf{SUMITMATPROD}$.
	\end{lemma}
	\begin{proof}
		Consider $\langle H,s,t,b \rangle \in \mathsf{MATINV}_{\kappa,\epsilon^{-1}}^{+}$. For $m \in \mathbb{N}$, we have $$\sum_{j=0}^m (I-H)^j=H^{-1}(I-(I-H)^{m+1}).$$ Let $\widehat{m}=\lceil \kappa(n) \rceil \lfloor 1+\log(\lfloor 4 \kappa(n) \epsilon(n)^{-1} \rfloor) \rfloor$. For $j \in [\widehat{m}]$, let $\widehat{A}_j=I_{jn} \oplus (I_{\widehat{m}-j+1} \otimes (I-H)) \in \widehat{\mat}(n\widehat{m}+n)$. For $1 \leq j_1 \leq j_2 \leq \widehat{m}$, we have $$\sigma_1(\widehat{A}_{j_1,j_2})=\sigma_1\left(I_{j_1 n} \oplus \left(\bigoplus_{k=1}^{\widehat{m}-j_1+1} (I-H)^{\min(k,j_2-j_1+1)}\right)\right)=\max_{k \in \{0,\ldots,j_2-j_1+1\}} \sigma_1((I-H)^k) = 1.$$ 
		Let $\widehat{E}=\{(s+jn,t+jn):j \in [\widehat{m}]\}$. We then have $$\sum_{(\widehat{s},\widehat{t}) \in \widehat{E}} \widehat{A}_{1,\widehat{m}}[\widehat{s},\widehat{t}]  = \sum_{j=0}^{\widehat{m}} (I-H)^j[s,t] = (H^{-1}(I-(I-H)^{\widehat{m}+1}))[s,t].$$ 
		This implies $$\left\lvert \bigg\lvert \sum_{(\widehat{s},\widehat{t}) \in \widehat{E}} \widehat{A}_{1,{\widehat{m}}}[\widehat{s},\widehat{t}]\bigg\rvert -\lvert H^{-1}[s,t]\rvert \right\rvert \leq \lvert (H^{-1}(I-H)^{\widehat{m}+1})[s,t] \rvert \leq \sigma_1(H^{-1}(I-H)^{\widehat{m}+1})$$ $$ \leq \sigma_1(H^{-1})(\sigma_1(I-H))^{\widehat{m}+1}\leq \kappa(n)\left(1-\frac{1}{\kappa(n)}\right)^{\widehat{m}+1}\leq \frac{1}{4}\epsilon(n).$$ Let $\widehat{b}=b-\frac{1}{4}\epsilon(n)$. We then conclude that $\langle \widehat{A}_1,\ldots,\widehat{A}_{\widehat{m}},\widehat{E},\widehat{b} \rangle \in \mathsf{SUMITMATPROD}_{n\widehat{m}+n,\widehat{m},1,2\epsilon^{-1}(n)}$ and $\mathsf{MATINV}(\langle H,s,t,b \rangle) = \mathsf{SUMITMATPROD}(\langle \widehat{A}_1,\ldots,\widehat{A}_{\widehat{m}},\widehat{E},\widehat{b} \rangle)$.
	\end{proof}
	
	\begin{lemma}\label{lemma:reduction:sumIterMatProdToIterMatProd}
		$\mathsf{SUMITMATPROD} \leq_{\mathsf{AC}^0}^m \mathsf{ITMATPROD}$.
	\end{lemma}
	\begin{proof}
		Consider $\langle A_1,\ldots,A_m,E,b \rangle \in \mathsf{SUMITMATPROD}_{n,m,\kappa,\epsilon^{-1}}$. Let $T_{c,d} \in \widehat{\mat}(n)$ denote the permutation matrix corresponding to interchanging elements $c,d \in [n]$ and leaving all other elements fixed. For $j \in [m]$, let $\widehat{A}_j=\bigoplus\limits_{(s,t) \in E} T_{1,t} A_j T_{1,s} \in \widehat{\mat}(n\lvert E \rvert).$ Let $R \in \widehat{\mat}(\lvert E \rvert)$ be defined such that $R_{r,c}=1$ if $r=c$ or $r=1$, and $R_{r,c}=0$ otherwise; let $\widehat{A}_0=R \otimes I_n$ and $\widehat{A}_{m+1}=\widehat{A}_0^{\dagger}$. We then have $\widehat{A}_{0,m+1}[1,1]=\sum\limits_{(s,t) \in E} A_{1,m}[s,t]$.
		Notice that $\sigma_1(\widehat{A}_0)=\sigma_1(\widehat{A}_{m+1})=\sigma_1(R)\sigma_1(I_n) \leq \sqrt{2\lvert E \rvert}$, which implies $$\sigma_1\left(\widehat{A}_{j_1,j_2}\right)\leq 2\lvert E \rvert \sigma_1\left(A_{\max(j_1,1),\min(j_2,m)}\right) \leq 2\lvert E \rvert \kappa(n) \leq 2 n^2 \kappa(n), \text{ for } 0 \leq j_1 \leq j_2 \leq m+1.$$  
		
		Let $\widehat{s}=\widehat{t}=1$ and $\widehat{b}=b$. Then $\langle \widehat{A}_0,\ldots,\widehat{A}_{m+1},\widehat{s},\widehat{t},\widehat{b} \rangle \in \mathsf{ITMATPROD}_{n^3,m+2,2n^2\kappa(n),\epsilon^{-1}(n)}$ and $\mathsf{SUMITMATPROD}(\langle A_1,\ldots,A_m,E,b \rangle) =\mathsf{ITMATPROD}(\langle \widehat{A}_0,\ldots,\widehat{A}_{m+1},\widehat{s},\widehat{t},\widehat{b} \rangle)$. 
	\end{proof}
	
	We now prove \Cref{thm:bqulComplete} from \Cref{sec:intro:linAlg}, which we restate here for convenience.
	
	\restateBqulComplete*
	\begin{proof}
		By \cite[Theorem 13]{fefferman2018complete}, $poly\text{-conditioned-}\mathsf{MATINV}$ is $\mathsf{BQ_U L}$-complete. By \Cref{lemma:reduction:iterMatProdToMatPow,lemma:reduction:matPowToMatInv,lemma:reduction:matInvToPosMatInv,lemma:reduction:posDetToSumIterMatProd,lemma:reduction:iterMatProdToNonnegIterMatProd,lemma:reduction:nonnegIterMatProdToDet,lemma:reduction:detToPosDet,lemma:reduction:posMatInvToSumIterMatProd,lemma:reduction:sumIterMatProdToIterMatProd}, we have $$\mathsf{MATINV}^{+} \leq_{\mathsf{AC}^0}^m \mathsf{SUMITMATPROD} \leq_{\mathsf{AC}^0}^m \mathsf{ITMATPROD} \leq_{\mathsf{AC}^0}^m \mathsf{MATPOW} \leq_{\mathsf{AC}^0}^m \mathsf{MATINV} \leq_{\mathsf{NC}^1}^m \mathsf{MATINV}^{+}$$ and $$\mathsf{DET}^{+} \leq_{\mathsf{AC}^0}^m \mathsf{SUMITMATPROD} \leq_{\mathsf{AC}^0}^m \mathsf{ITMATPROD} \leq_{\mathsf{AC}^0}^m \mathsf{ITMATPROD}^{\geq 0} \leq_{\mathsf{AC}^0}^m \mathsf{DET} \leq_{\mathsf{NC}^1}^m \mathsf{DET}^{+}.$$ Recall that $\leq_{\mathsf{AC}^0}^m$ or $\leq_{\mathsf{NC}^1}^m$ reducibility implies $\leq_{\mathsf{L}}^m$ reducibility, and $\mathcal{P}  \leq_{\mathsf{L}}^m \mathcal{P}' \Rightarrow poly\text{-conditioned-}\mathcal{P}  \leq_{\mathsf{L}}^m poly\text{-conditioned-}\mathcal{P}'$. Therefore, each such $poly\text{-conditioned-}\mathcal{P}$ is $\mathsf{BQ_U L}$-complete.
	\end{proof}
	
	\section{Fully Logarithmic Approximation Schemes}\label{sec:determinant:approximateCounting}
	
	We next study the class of functions that are well-approximable in quantum logspace, following (essentially) the notation and definitions of \cite{doron2015randomization}. In particular, we work with the general (resp. unitary) quantum Turing machine model, rather than the \textit{equivalent} model of a uniform family of general (resp. unitary) quantum circuits; of course, all results also apply to the quantum circuit model. For simplicity, throughout this section, we fix the alphabet $\Sigma=\{0,1\}$. We say that a function $f:\Sigma^* \rightarrow \mathbb{R}$ is $poly$-bounded if $\lvert f(w) \rvert \leq poly(\lvert w \rvert)$, $\forall w \in \Sigma^*$.
	
	\begin{definition}\label{def:fullyLogApproxScheme}
		We say that a $poly$-bounded $f$ has a \textit{fully logarithmic quantum approximation scheme} $\mathsf{FLQAS}$ if there is a (general) quantum TM $M_f$ that, on input $\langle x,\epsilon,\delta \rangle$, where $x \in \Sigma^*$ and $\epsilon,\delta \in \mathbb{R}_{>0}$, runs in time $poly(\lvert x \rvert, \epsilon^{-1},\log(\delta^{-1}))$ and space $O(\log(\lvert x \rvert)+\log(\epsilon^{-1})+\log(\log(\delta^{-1})))$, and outputs a value $y \in \mathbb{R}$ such that $\Pr[\lvert f(x)-y \rvert \geq \epsilon] \leq \delta$ (to be precise, $M_f$ outputs a string that encodes a dyadic rational number $y$). In other words, with \textit{confidence} at least $1-\delta$, the value $y$ is an \textit{additive} approximation of $f(x)$ with \textit{error} at most $\epsilon$. We analogously say that $f$ has a $\mathsf{FLQ_U AS}$ if $M_f$ is a \textit{unitary} quantum TM, a $\mathsf{FLRAS}$ if $M_f$ is a \textit{randomized} TM, and a $\mathsf{FLAS}$ if $M_f$ is a \textit{deterministic} TM (where, in this last case, we set $\delta=0$ and remove the dependence on $\delta$ from the time and space bounds).
	\end{definition} 
	
	Following the notation established in \Cref{sec:prelim:notation} and \Cref{def:detPromiseProblems:detAndInv,def:detPromiseProblems:powAndItprod,def:polyConditionedDetPromiseProblems}, let $$\mathcal{D}(poly\text{-}matinv)=\bigcup_n \{\langle A,s,t \rangle:A \in \widehat{\mat}(n,n^{-O(1)},1), s,t \in [n] \}.$$ In other words, $\mathcal{D}(poly\text{-}matinv)$ consists of encodings of instances of a variant of $poly\text{-conditioned-}\mathsf{MATINV}$ where we only have a promise involving the singular values (i.e., no restriction on $A^{-1}[s,t]$ involving $b$). We then consider the $poly$-conditioned matrix inversion function $\lvert poly\text{-}matinv(\cdot) \rvert:\mathcal{D}(poly\text{-}matinv)\rightarrow \mathbb{R}_{\geq 0}$, which is given by $\lvert poly\text{-}matinv(\langle A,s,t \rangle)\rvert=\lvert A^{-1}[s,t]\rvert$. 
	
	Similarly, $\mathcal{D}(poly\text{-}itmatprod)$ consists of all $\langle A_1,\ldots,A_m,s,t \rangle$, where $m=poly(n)$, $A_1,\ldots,A_m \in \widehat{\mat}(n)$, $\sigma_1(A_{j_1,j_2}) \leq poly(n)$ for $1 \leq j_1 \leq j_2 \leq m$, and $s,t \in [n]$. We then define the function $\lvert poly\text{-}itmatprod(\cdot) \rvert:\mathcal{D}(poly\text{-}itmatprod) \rightarrow \mathbb{R}_{\geq 0}$ by $\lvert poly\text{-}itmatprod(\langle A_1,\ldots,A_m,s,t \rangle)\rvert=\lvert A_{1,m}[s,t] \rvert$. Lastly, we define $\mathcal{D}(poly\text{-}det)=\bigcup_n \{\langle A \rangle:A \in \widehat{\mat}(n,n^{-O(1)},1)\}$. Note that the promise problem $\mathsf{DET}$, given in \Cref{def:detPromiseProblems:detAndInv}, corresponds to approximating the function $\ln(\lvert poly\text{-}det(\cdot)\rvert):\mathcal{D}(poly\text{-}det) \rightarrow \mathbb{R}_{\leq 0}$, defined in the obvious way.
	
	Note that, following \cite{doron2015randomization}, we have defined fully logarithmic (quantum, randomized, etc.) approximation schemes with respect to \textit{additive} error $\epsilon$; that is to say, we approximate $f(x)$ by a value $y$ such that $\Pr[\lvert f(x)-y \rvert \geq \epsilon] \leq \delta$. We then define a \textit{multiplicative} fully logarithmic (quantum, randomized, etc.) approximation scheme of a function $g:\Sigma^* \rightarrow \mathbb{R}_{\geq 0}$ as an analogous procedure that produces an approximation $z$ such that $\Pr[z \not \in [e^{-\epsilon}g(x),e^{\epsilon} g(x)]] \leq \delta$. Note that here, for convenience, we follow the convention (as used in, for example, \cite{jerrum2004polynomial}) that multiplicative approximations are defined using $e^{\pm \epsilon}$, rather than the more standard (and essentially equivalent) $(1 \pm \epsilon)$. Note that $\ln(\lvert poly\text{-}det(\cdot)\rvert)$ has an (additive) $\mathsf{FLQ_U AS}$ (resp. $\mathsf{FLQAS},\mathsf{FLRAS},\mathsf{FLAS}$) if and only if $\lvert poly\text{-}det(\cdot)\rvert$ has a multiplicative $\mathsf{FLQ_U AS}$ (resp. $\mathsf{FLQAS},\mathsf{FLRAS},\mathsf{FLAS}$); this follows from the fact that $\lvert \ln(\lvert \det(A) \rvert)-y \rvert \geq \epsilon$ if and only if $e^y \not \in [e^{-\epsilon}\lvert \det(A) \rvert,e^{\epsilon}\lvert \det(A) \rvert]$.
	
	The function $\lvert poly\text{-}matinv(\cdot) \rvert$ is known to have a $\mathsf{FLQAS}$ \cite{doron2015randomization}. We improve on this result.
	
	\begin{lemma}\label{lemma:polyCondDetCompleteProblemsHaveFlquas}
		Each of $\lvert poly\text{-}matinv(\cdot) \rvert$, $\lvert poly\text{-}itmatprod(\cdot) \rvert$, and $\ln(\lvert poly\text{-}det(\cdot)\rvert)$ have a $\mathsf{FLQ_U AS}$. Moreover, $\lvert poly\text{-}det(\cdot)\rvert$ has a multiplicative $\mathsf{FLQ_U AS}$.
	\end{lemma}
	\begin{proof}
		By \cite[Theorem 14]{fefferman2018complete} (and the discussion following it), $\lvert poly\text{-}matinv(\cdot) \rvert$ has a $\mathsf{FLQ_U AS}$. By \Cref{lemma:reduction:iterMatProdToMatPow,lemma:reduction:matPowToMatInv}, $\lvert poly\text{-}itmatprod(\cdot) \rvert$ has a $\mathsf{FLQ_U AS}$. Finally, by \Cref{lemma:reduction:detToPosDet,lemma:reduction:posDetToSumIterMatProd,lemma:reduction:sumIterMatProdToIterMatProd}, $\ln(\lvert poly\text{-}det(\cdot)\rvert)$ has a $\mathsf{FLQ_U AS}$; this implies $\lvert poly\text{-}det(\cdot)\rvert$ has a multiplicative $\mathsf{FLQ_U AS}$.
	\end{proof}
	
	Doron and Ta-Shma \cite[Theorem 6]{doron2015randomization} showed that, if $\mathsf{BQL}=\mathsf{BPL}$, then every $poly$-bounded function that has a $\mathsf{FLQAS}$ also has a $\mathsf{FLRAS}$ (recall that we use $\mathsf{BQL}$ and $\mathsf{BPL}$ to denote classes of promise problems, which differs from the notation used in \cite{doron2015randomization}). By combining this with the $\mathsf{BQ_U L}$-hardness of the various $poly$-conditioned promise problems (\Cref{thm:bqulComplete}) and our result that $\mathsf{BQ_U L}=\mathsf{BQL}$ (\Cref{thm:bqlEqualsBqulEqualsQmal}), the following proposition is immediate; we note that a partial (weaker) version of this proposition was implicit in \cite{doron2017approximating}.
	
	\begin{proposition}\label{thm:bplEqualsBqlIffApproxSchemesCoincide}
		The following statements are equivalent.
		\begin{enumerate}[(i)]
			\item $\mathsf{BQL}=\mathsf{BPL}$.
			\item Every $poly$-bounded function that has a $\mathsf{FLQAS}$ also has a $\mathsf{FLRAS}$.
			\item Every $poly$-bounded function that has a $\mathsf{FLQ_U AS}$ also has a $\mathsf{FLRAS}$.
			\item $\lvert poly\text{-}matinv(\cdot) \rvert$ has a $\mathsf{FLRAS}$.
			\item $\lvert poly\text{-}itmatprod(\cdot) \rvert$ has a $\mathsf{FLRAS}$.
			\item $\ln(\lvert poly\text{-}det(\cdot)\rvert)$ has a $\mathsf{FLRAS}$.
			\item $\lvert poly\text{-}det(\cdot)\rvert$ has a multiplicative $\mathsf{FLRAS}$.
		\end{enumerate}
	\end{proposition}
	
	\begin{remark}
		The preceding proposition suggests that $\lvert poly\text{-}det(\cdot)\rvert$ \textit{does not} have a multiplicative $\mathsf{FLRAS}$, as this would imply the \textit{seemingly unlikely} statement $\mathsf{BQL}=\mathsf{BPL}$. It is natural to compare this with the result of Jerrum, Sinclair, and Vigoda \cite{jerrum2004polynomial} which shows the existence of a multiplicative FPRAS (fully \textit{polynomial} randomized approximation scheme) for the \textit{permanent} of a nonnegative integer matrix.
	\end{remark}
	
	\section{Well-Conditioned Singular}\label{sec:singular}
	
	The class $\mathsf{C_{=}L}$ has a collection of natural complete problems, given by the ``verification" versions of the standard $\mathsf{DET}^{*}$-complete problems \cite{santha1998verifying}. We now study the well-conditioned versions of these problems.
	
	\begin{definition}\label{def:singularPromiseProblems}
		Consider functions $m:\mathbb{N} \rightarrow \mathbb{N}$, $\kappa:\mathbb{N} \rightarrow \mathbb{R}_{\geq 1}$, and $\epsilon:\mathbb{N} \rightarrow \mathbb{R}_{>0}$. 
		\newline$\mathsf{SINGULAR}_{n,\epsilon^{-1}}$
		\newline\hspace*{5pt}\textit{Input}: $A \in \widehat{\herm}(n)$
		\newline\hspace*{5pt}\textit{Promise}: $\sigma_1(A) \leq 1$, $\sigma_n(A) \in \{0\} \cup [\epsilon(n),1]$
		\newline\hspace*{5pt}\textit{Output}: $1$ if $\sigma_n(A)=0$, $0$ otherwise
		\newline $\mathsf{vMATINV}_{n,\kappa,\epsilon^{-1}}$
		\newline\hspace*{5pt}\textit{Input}: $A \in \widehat{\mat}(n)$, $s,t \in [n]$, $b \in \mathbb{Q}[i]_n$
		\newline\hspace*{5pt}\textit{Promise}: $A \in \widehat{\mat}(n,1/\kappa(n),1)$, $\lvert A^{-1}[s,t]-b \rvert \in \{0\} \cup [\epsilon(n),2\kappa(n)]$
		\newline\hspace*{5pt}\textit{Output}: $1$ if $ A^{-1}[s,t] = b$, $0$ otherwise	
		\newline $\mathsf{vMATPOW}_{n,m,\kappa,\epsilon^{-1}}$
		\newline\hspace*{5pt}\textit{Input}: $A \in \widehat{\mat}(n)$, $s,t \in [n]$, $b \in \mathbb{Q}[i]_n$
		\newline\hspace*{5pt}\textit{Promise}: $\sigma_1(A^j) \leq \kappa(n) \, \forall j \in [m]$, $\lvert A^m[s,t]-b \rvert \in \{0\} \cup [\epsilon(n),2\kappa(n)]$
		\newline\hspace*{5pt}\textit{Output}: $1$ if $ A^m[s,t]= b$, $0$ otherwise
		\newline $\mathsf{vITMATPROD}_{n,m,\kappa,\epsilon^{-1}}$
		\newline\hspace*{5pt}\textit{Input}: $A_1,\ldots,A_m \in \widehat{\mat}(n)$, $s,t \in [n]$, $b \in \mathbb{Q}[i]_n$ 
		\newline\hspace*{5pt}\textit{Promise}: $\sigma_1(A_{j_1,j_2}) \leq \kappa(n) \, \forall 1 \leq j_1 \leq j_2 \leq m$, $\left\lvert A_{1,m}[s,t] - b \right\rvert \in \{0\} \cup [\epsilon(n),2\kappa(n)]$
		\newline\hspace*{5pt}\textit{Output}: $1$ if $A_{1,m}[s,t] = b$, $0$ otherwise
	\end{definition} 

	We begin by exhibiting reductions between the above problems; in subsequent sections, we will use these reductions to prove new properties of quantum logspace.
	
	\begin{lemma}\label{lemma:reduction:vMatInvToSingular}
		$\mathsf{vMATINV}\leq_{\mathsf{AC}^0}^m \mathsf{SINGULAR}$.
	\end{lemma}
	\begin{proof}
		Consider $\langle A,s,t,b \rangle \in \mathsf{vMATINV}_{n,\kappa,\epsilon^{-1}}$. We define $\widehat{B}=(2\lceil \kappa(n) \rceil A) \oplus \left(1-\frac{b}{2\lceil \kappa(n) \rceil}\right)^{-1}I_1 \in \widehat{\mat}(n+1)$, $u=\ket{s}+\ket{n+1}$, $v=\ket{t}+\ket{n+1}$, and $\widehat{C}=\widehat{B}-v u^{\dagger} \in \widehat{\mat}(n+1)$. By the matrix determinant lemma, $$\det(\widehat{C})=(1-u \widehat{B}^{-1} v)\det(\widehat{B})=\left(1-\left(\frac{A^{-1}[s,t]}{2\lceil \kappa(n) \rceil}+\left(1-\frac{b}{2\lceil \kappa(n) \rceil}\right)\right)\right)\det(\widehat{B})=\frac{b-A^{-1}[s,t]}{2\lceil \kappa(n) \rceil}\det(\widehat{B}).$$
		
		If $A^{-1}[s,t]=b$, then $\det(\widehat{C})=0$, which implies $\sigma_{n+1}(\widehat{C})=0$. Next, suppose $\lvert A^{-1}[s,t]-b \rvert \geq \epsilon(n)$, then $\lvert \det(\widehat{C}) \rvert \geq \frac{\epsilon(n)}{2\lceil \kappa(n) \rceil} \lvert \det(\widehat{B}) \rvert$. By the Weyl inequalities, $\sigma_1(\widehat{C})\leq \sigma_1(\widehat{B})+\sigma_1(-v u^{\dagger})=\sigma_1(\widehat{B})+1$ and, for $j \in [n+1]$, we have $\sigma_j(\widehat{C}) \leq \sigma_{j-1}(\widehat{B})+\sigma_2(-v u^{\dagger})=\sigma_{j-1}(\widehat{B})$. Moreover, $\sigma_1(\widehat{B})=2\lceil \kappa(n) \rceil \sigma_1(A)\leq 2\lceil \kappa(n) \rceil$, $\sigma_n(\widehat{B})=2\lceil \kappa(n) \rceil \sigma_n(A) \geq 2$, and $\sigma_{n+1}(\widehat{B})=\left\lvert1-\frac{b}{2\lceil \kappa(n) \rceil}\right\rvert^{-1} \geq \frac{2}{\sqrt{5}}$. Therefore, $$\sigma_{n+1}(\widehat{C})=\frac{\lvert \det(\widehat{C}) \rvert}{\sigma_1(\widehat{C}) \prod\limits_{j=2}^n \sigma_j(\widehat{C})} \geq \frac{\frac{\epsilon(n)}{2\lceil \kappa(n) \rceil} \lvert \det(\widehat{B}) \rvert}{(\sigma_1(\widehat{B})+1) \prod\limits_{j=1}^{n-1} \sigma_j(\widehat{B})}=\frac{\epsilon(n)\sigma_n(\widehat{B})\sigma_{n+1}(\widehat{B})}{2\lceil \kappa(n) \rceil (\sigma_1(\widehat{B})+1)} \geq \frac{2\epsilon(n)}{\sqrt{5} \lceil \kappa(n) \rceil (2\lceil \kappa(n) \rceil+1)}.$$  
		
		Let $\widehat{d}=\frac{1}{2\lceil \kappa(n) \rceil+1}$ and $\widehat{H}=\widehat{d}\begin{pmatrix}
		0_{n+1} & \widehat{C} \\
		\widehat{C}^{\dagger} & 0_{n+1}
		\end{pmatrix}\in \widehat{\herm}(2n+2)$. Notice that $\widehat{H}$ has eigenvalues $\{\pm \widehat{d}\sigma_1(\widehat{C}),\ldots,\pm \widehat{d}\sigma_{n+1}(\widehat{C})\}$. This implies $\sigma_1(\widehat{H})=\widehat{d}\sigma_1(\widehat{C}) \leq 1$ and $\sigma_{2n+2}(\widehat{H})=\widehat{d}\sigma_{n+1}(\widehat{C}) \in \{0\} \cup \left[\frac{2\epsilon(n)}{\sqrt{5} \lceil \kappa(n) \rceil (2\lceil \kappa(n) \rceil+1)^2},1\right]$. Moreover, $\sigma_{2n+2}(\widehat{H})=0 \Leftrightarrow A^{-1}[s,t]=b$. Therefore, $\mathsf{vMATINV}(\langle A,s,t,b \rangle)=\mathsf{SINGULAR}(\langle \widehat{H} \rangle)$ and $\langle \widehat{H} \rangle \in \mathsf{SINGULAR}_{2n+2,(2\epsilon(n))^{-1}\sqrt{5} \lceil \kappa(n) \rceil (2\lceil \kappa(n) \rceil+1)^2}$.
	\end{proof}
	
	The following lemmas have proofs analogous to those of \Cref{lemma:reduction:iterMatProdToMatPow} and \Cref{lemma:reduction:matPowToMatInv}, respectively.
	
	\begin{lemma}\label{lemma:reduction:vIterMatProdToVMatPow}
		$\mathsf{vITMATPROD} \leq_{\mathsf{AC}^0}^m \mathsf{vMATPOW}$.
	\end{lemma}
	
	\begin{lemma}\label{lemma:reduction:vMatPowToVMatInv}
		$\mathsf{vMATPOW} \leq_{\mathsf{AC}^0}^m \mathsf{vMATINV}$.
	\end{lemma}
	
	\subsection{RQSPACE vs. R\texorpdfstring{Q\textsubscript{U}}{QU}SPACE vs. R\texorpdfstring{Q\textsubscript{U}}{QU}MASPACE }\label{sec:singular:rqlQma}
	
	In this section, we will prove the following relationships between the various one-sided bounded-error space-bounded quantum complexity classes.
	
	\begin{theorem}\label{thm:rquspaceVsRqspaceVsRqmaspaceEtc}
		For any space-constructible function $s: \mathbb{N} \rightarrow \mathbb{N}$, where $s(n)=\Omega(\log n)$, we have $$\mathsf{RQ_UMASPACE}(s)=\mathsf{RQ_U SPACE}(s) \subseteq \mathsf{RQSPACE}(s) \subseteq \mathsf{coQ_U MASPACE_1}(s).$$
	\end{theorem}
	
	This theorem, which is very much the one-sided error analogue of \Cref{thm:intro:bquspaceEqualsBqspace}, has a proof which follows the same general structure. For those promise problems $\mathcal{P}$ given in \Cref{def:singularPromiseProblems}, we define $poly\text{-conditioned-}\mathcal{P}$ as in \Cref{def:polyConditionedDetPromiseProblems}. 
	
	\begin{lemma}\label{lemma:vIterMatPowIsCorqlHard}
		$poly\text{-conditioned-}\mathsf{vITMATPROD}$ is $\mathsf{coRQL}$-hard.
	\end{lemma}
	\begin{proof}
		Precisely analogous to the proof of \Cref{lemma:itMatProdIsBqlHard}.
	\end{proof}
	
	\begin{lemma}\label{lemma:rqmalInRqul}
		$\mathsf{RQ_U MAL} \subseteq \mathsf{RQ_U L}$
	\end{lemma}
	\begin{proof}[Proof (sketch)]
		Apply the well-known technique of replacing Merlin's proof with the totally mixed state \cite{marriott2005quantum}, which preserves perfect soundness \cite{kobayashi2003quantum}; then use space-efficient probability amplification for one-sided bounded-error (unitary) quantum logspace \cite{watrous2001quantum}. We present the details in \Cref{sec:appendix:proofOfRqumalInRqul}.
	\end{proof}
	
	\begin{lemma}\label{lemma:singularIsQMAL1Complete}
		$poly\text{-conditioned-}\mathsf{SINGULAR}$ is $\mathsf{Q_U MAL_1}$-complete.
	\end{lemma}
	\begin{proof}
		We first show $\mathsf{Q_U MAL_1}$-hardness. Suppose $\mathcal{P}=(\mathcal{P}_1,\mathcal{P}_0) \in \mathsf{Q_U MAL_1}$. By definition, there is a uniform family of (unitary) quantum circuits $\{V_w=(V_{w,1},\ldots,V_{w,t_w}):w \in \mathcal{P}\}$, where $V_w$ acts on $m_w+h_w=O(\log \lvert w \rvert)$ qubits and has $t_w=poly(\lvert w \rvert)$ gates, such that  $w \in \mathcal{P}_1 \Rightarrow \exists \ket{\psi} \in \Psi_{m_w}, \Pr[V_w \text{ accepts } w,\ket{\psi}] \geq c= 1$, and $w \in \mathcal{P}_0 \Rightarrow \forall \ket{\psi} \in \Psi_{m_w}, \Pr[V_w \text{ accepts } w,\ket{\psi}]\leq k=\frac{1}{2}$, where $\Pr[V_w \text{ accepts } w,\ket{\psi}]= \lVert \Pi_1 V_w (\ket{\psi} \otimes \ket{0^{h_w}}) \rVert^2$. 
		
		We make use of the Kitaev clock Hamiltonian construction \cite[Section 14.4]{kitaev1997quantum}, in a manner similar to \cite[Lemma 21]{fefferman2018complete} (though, without the need to first apply space-efficient probability amplification techniques). Let $d_w=2^{m_w +h_w}(t_w+1)=poly(\lvert w \rvert)$, define the $d_w$-dimensional Hilbert space $\mathcal{H}_w=\mathbb{C}^{2^{m_w}} \otimes \mathbb{C}^{2^{h_w}} \otimes \mathbb{C}^{t_w+1}$, and let $\Pi_b=I_{2^{b-1}} \otimes \ket{1}\bra{1} \otimes I_{2^{m_w +h_w-b}}\in \widehat{\proj}(2^{m_w+h_w})$ denote the projection onto the subspace of $\mathbb{C}^{2^{m_w}} \otimes \mathbb{C}^{2^{h_w}}$ spanned by states in which the $b^{\text{th}}$ qubit is $1$. We define the Hamiltonians $H_w^{prop},H_w^{in},H_w^{out},H_w \in \widehat{\pos}(d_w)$ on $\mathcal{H}_w$ as follows: 
		$$H_w^{prop}=\frac{1}{2}\sum_{j=1}^{t_w} \left(-V_{w,j} \otimes \ket{j}\bra{j-1}-V_{w,j}^{\dagger} \otimes \ket{j-1}\bra{j}+I_{2^{m_w +h_w}} \otimes (\ket{j}\bra{j}+\ket{j-1}\bra{j-1}) \right)$$
		$$H_w^{in}=\sum_{b=m_w+1}^{m_w+h_w}\Pi_b \otimes \ket{0}\bra{0}, \ \  H_w^{out}=\Pi_1 \otimes \ket{t_w}\bra{t_w}, \ \ \text{and} \ \ H_w=H_w^{in}+H_w^{prop}+H_w^{out}.$$
		
		By \cite[Section 14.4]{kitaev1997quantum}, $\exists r_0,r_1 \in \mathbb{R}_{>0}$, such that $\forall w \in \mathcal{P}$ we have: $\sigma_1(H_w) \leq r_0$, $w \in \mathcal{P}_1 \Rightarrow \sigma_{d_w}(H_w) \leq \frac{1-c}{t_w+1}=0$, and $w \in \mathcal{P}_0 \Rightarrow \sigma_{d_w}(H_w) \geq r_1 \frac{1-\sqrt{k}}{t_w+1}=1/poly(d_w)$. Therefore, $\mathcal{P}(w)=\mathsf{SINGULAR}(\langle H_w \rangle)$ and $\langle H_w \rangle \in poly\text{-conditioned-}\mathsf{SINGULAR}$, which implies $poly\text{-conditioned-}\mathsf{SINGULAR}$ is $\mathsf{Q_U MAL_1}$-hard.
		
		Finally, $poly\text{-conditioned-}\mathsf{SINGULAR} \in \mathsf{Q_U MAL_1}$ follows from using the quantum walk based Hamiltonian simulation technique of Childs \cite{childs2010relationship,berry2014exponential} to allow the phase estimation of \cite[Lemma 19]{fefferman2018complete} to be carried out with one-sided error, in the style of \cite[Proposition 32]{fefferman2018complete}, we omit the straightforward details.
	\end{proof}
	
	\begin{lemma}\label{thm:rqlVsRqulVsQmal1Etc}
		$\mathsf{RQ_U MAL} = \mathsf{RQ_U L} \subseteq \mathsf{RQL} \subseteq \mathsf{coQ_U MAL_1}$.
	\end{lemma}
	\begin{proof}
		Clearly, $\mathsf{RQ_U L} \subseteq \mathsf{RQ_U MAL}$. By \Cref{lemma:rqmalInRqul}, $\mathsf{RQ_U MAL} \subseteq \mathsf{RQ_U L}$, which implies $\mathsf{RQ_U MAL} = \mathsf{RQ_U L}$. Trivially, $\mathsf{RQ_U L} \subseteq \mathsf{RQL}$ By \Cref{lemma:vIterMatPowIsCorqlHard}, $poly\text{-conditioned-}\mathsf{vITMATPROD}$ is $\mathsf{coRQL}$-hard; thus, by  \Cref{lemma:reduction:vIterMatProdToVMatPow,lemma:reduction:vMatPowToVMatInv,lemma:reduction:vMatInvToSingular}, $poly\text{-conditioned-}\mathsf{SINGULAR}$ is $\mathsf{coRQL}$-hard. Finally, by \Cref{lemma:singularIsQMAL1Complete}, we have $poly\text{-conditioned-}\mathsf{SINGULAR} \in \mathsf{Q_U MAL_1}$, which implies $\mathsf{coRQL} \subseteq \mathsf{Q_U MAL_1}$; thus, $\mathsf{RQL} \subseteq \mathsf{coQ_U MAL_1}$.
	\end{proof}
	
	The main theorem stated at the beginning of this section now follows immediately.
	
	\begin{proof}[Proof of \Cref{thm:rquspaceVsRqspaceVsRqmaspaceEtc}]
		Follows from \Cref{thm:rqlVsRqulVsQmal1Etc} and a padding argument analogous to \Cref{thm:intro:bquspaceEqualsBqspace}.
	\end{proof}
	
	\subsection{NQSPACE vs. N\texorpdfstring{Q\textsubscript{U}}{QU}SPACE vs. N\texorpdfstring{Q\textsubscript{U}}{QU}SPACE}\label{sec:singular:nqlQma}
	
	We next consider variants of the well-conditioned problems of \Cref{def:singularPromiseProblems}, in which $\epsilon(n)=0 \, \forall n \in \mathbb{N}$ and $\kappa(n)=poly(n)$, which we denote by $\mathsf{PreciseSINGULAR}$, $\mathsf{vPreciseITMATPROD}$, etc. By arguments precisely analogous to those of \Cref{sec:singular:rqlQma}, we obtain the following three lemmas; we omit the proofs.
	
	\begin{lemma}\label{lemma:vPreciseMatPowIsConqlHard}
		$poly\text{-conditioned-}\mathsf{vPreciseITMATPROD}$ is $\mathsf{coNQL}$-hard.
	\end{lemma}
	
	\begin{lemma}\label{lemma:nqmalInNqul}
		$\mathsf{NQ_U MAL} \subseteq \mathsf{NQ_U L}$
	\end{lemma}
	
	\begin{lemma}\label{lemma:preciseSingularIsPreciseQMAL1Complete}
		$\mathsf{PreciseSINGULAR}$ is $\mathsf{PreciseQ_U MAL_1}$-complete.
	\end{lemma}
	
	We then obtain the following relationships between the various one-sided unbounded-error classes.
	
	\begin{lemma}\label{thm:nqlEqualsNqulEqualsNqmalEqualsEtc}
		$\mathsf{NQ_U MAL}=\mathsf{NQ_U L}=\mathsf{NQL}=\mathsf{coPreciseQ_U MAL_1}=\mathsf{coC_{=}L}$.
	\end{lemma}
	\begin{proof}
		Trivially, $\mathsf{NQ_U L} \subseteq \mathsf{NQL}$ and $\mathsf{NQ_U L} \subseteq \mathsf{NQ_U MAL}$. By \Cref{lemma:nqmalInNqul}, $\mathsf{NQ_U MAL} \subseteq \mathsf{NQ_U L}$, which implies $\mathsf{NQ_U MAL} = \mathsf{NQ_U L}$. By \Cref{lemma:vPreciseMatPowIsConqlHard}, $poly\text{-conditioned-}\mathsf{vPreciseITMATPROD}$ is $\mathsf{coNQL}$-hard; thus, by  \Cref{lemma:reduction:vIterMatProdToVMatPow,lemma:reduction:vMatPowToVMatInv,lemma:reduction:vMatInvToSingular}, $\mathsf{PreciseSINGULAR}$ is $\mathsf{coNQL}$-hard. By \Cref{lemma:preciseSingularIsPreciseQMAL1Complete}, $poly\text{-conditioned-}\mathsf{SINGULAR} \in \mathsf{PreciseQ_U MAL_1}$, which implies $\mathsf{NQL} \subseteq \mathsf{coPreciseQ_U MAL_1}$. By \cite[Theorem 14]{allender1996relationships}, $\mathsf{PreciseSINGULAR}$ is $\mathsf{C_{=}L}$-complete, and by \Cref{lemma:preciseSingularIsPreciseQMAL1Complete}, $\mathsf{PreciseSINGULAR}$ is $\mathsf{PreciseQ_U MAL_1}$-complete; therefore, $\mathsf{C_{=}L}=\mathsf{PreciseQ_U MAL_1}$. Thus far, we have shown $\mathsf{NQ_U MAL}=\mathsf{NQ_U L} \subseteq \mathsf{NQL} \subseteq \mathsf{coPreciseQ_U MAL_1}=\mathsf{coC_{=}L}$. To complete the proof, note that, by \cite[Theorem 4.14]{watrous1999space}, $\mathsf{NQ_U L}=\mathsf{coC_{=}L}$.
	\end{proof}
		
	\begin{theorem}\label{thm:nquspaceEqualsNqspaceEqualsNqmaspaceEtc}
	For any space-constructible function $s: \mathbb{N} \rightarrow \mathbb{N}$, where $s(n)=\Omega(\log n)$, we have $$\mathsf{NQ_U MASPACE}(s)=\mathsf{NQ_U SPACE}(s)=\mathsf{NQSPACE}(s)=\mathsf{coPreciseQ_U MA_1 SPACE}(s)=\mathsf{coC_{=}SPACE}(s).$$
\end{theorem}
	\begin{proof}
		Follows from \Cref{thm:nqlEqualsNqulEqualsNqmalEqualsEtc} and a padding argument analogous to that of \Cref{thm:intro:bquspaceEqualsBqspace}.
	\end{proof}
	
	\section{Discussion}\label{sec:discussion}
	
	We conclude by stating a few interesting open problems related to our work. In \Cref{thm:intro:bquspaceEqualsBqspace} we established the equivalence of unitary quantum space, general quantum space, and space-bounded quantum Merlin-Arthur proof systems, in the two-sided bounded-error case. We obtained an analogous equivalence for one-sided unbounded-error in \Cref{thm:nquspaceEqualsNqspaceEqualsNqmaspaceEtc}. However, in the case of one-sided bounded-error, we only have the partial results of \Cref{thm:rquspaceVsRqspaceVsRqmaspaceEtc}. In particular, specializing to the case of logspace, we have $\mathsf{BQL} = \mathsf{BQ_U L} = \mathsf{Q_U MAL}$ in the two-sided bounded-error case (\Cref{thm:bqlEqualsBqulEqualsQmal}), and we have 	$\mathsf{RQ_U MAL} = \mathsf{RQ_U L} \subseteq \mathsf{RQL} \subseteq \mathsf{coQ_U MAL_1}$ in the one-sided bounded-error case (\Cref{thm:rqlVsRqulVsQmal1Etc}). It is naturally to ask if the analogues of results known to hold for two-sided bounded-error also hold for one-sided bounded-error.
	
	\begin{openProb} 
		Is $\mathsf{RQ_U L} = \mathsf{RQL}$? Is $\mathsf{RQL} = \mathsf{coQ_U MAL_1}$? 
	\end{openProb}
	
	By the well-known result of Zachos and F\"{u}rer \cite{zachos1987probabilistic}, $\mathsf{MA}=\mathsf{MA_1}$; that is to say, it is possible to achieve perfect completeness for \textit{classical} (polynomial time) Merlin-Arthur proof systems. On the other hand, the question of whether or not it is possible to achieve perfect completeness for \textit{quantum} (polynomial time) Merlin-Arthur proof systems (i.e., is $\mathsf{QMA}=\mathsf{QMA_1}$?) remains open (see, for instance, \cite{aaronson2009perfect,aharonov2002quantum,bravyi2011efficient,jordan2012achieving} for previous discussion). We next consider the logspace analogue of this question.
	
	\begin{openProb} 
		Is $\mathsf{QMAL}=\mathsf{QMAL_1}$?
	\end{openProb}
	
	A possible explanation for the difficulty of proving $\mathsf{QMA}=\mathsf{QMA_1}$ (if these classes are indeed equal) was provided by Aaronson's result \cite{aaronson2009perfect} that there is a \textit{quantum} oracle $\mathcal{U}$ such that $\mathsf{QMA}^{\mathcal{U}} \neq \mathsf{QMA_1}^{\mathcal{U}}$; therefore, any proof of $\mathsf{QMA}=\mathsf{QMA_1}$ must use a technique that is \textit{quantumly nonrelativizing}. Note that the technique used by Zachos and F\"{u}rer \cite{zachos1987probabilistic} to show $\mathsf{MA}=\mathsf{MA_1}$ is (classically) relativizing. It is not hard to see that Aaronson's argument can also be used to produce a quantum oracle $\mathcal{U}$ such that $\mathsf{QMAL}^{\mathcal{U}} \neq \mathsf{QMAL_1}^{\mathcal{U}}$, and so any proof of $\mathsf{QMAL}=\mathsf{QMAL_1}$ must also use quantumly nonrelativizing techniques. We emphasize that the techniques used in this paper to show our results concerning new inclusions between various complexity classes (i.e., the various reductions between  linear-algebraic problems shown in this paper) are \textit{quantumly nonrelativizing}. Moreover, it is known that it is possible to achieve perfect completeness in quantum Merlin-Arthur proof systems that have a \textit{classical witness}; that is to say, $\mathsf{QCMA}=\mathsf{QCMA_1}$ \cite{jordan2012achieving}. Note that, trivially, $\mathsf{BQ_U L} \subseteq \mathsf{QCMAL} \subseteq \mathsf{QMAL}$. Thus, the known equality $\mathsf{BQ_U L} = \mathsf{QMAL}$ immediately implies $\mathsf{QCMAL}=\mathsf{QMAL}$. Therefore, $\mathsf{QMAL}=\mathsf{QMAL_1} \Leftrightarrow \mathsf{QCMAL}=\mathsf{QMAL_1} \Leftarrow \mathsf{QCMAL}=\mathsf{QCMAL_1}$.
	
 	We conclude with a general question.
	
	\begin{openProb}
		What further relationships can be established between $\mathsf{BQL}$ and other natural logspace complexity classes (e.g., $\mathsf{\# L}, \mathsf{GapL}, \mathsf{L}/poly$, etc.)?
	\end{openProb}
	
	\section*{Acknowledgments}
	
	B.F. and Z.R acknowledge support from AFOSR (YIP number
	FA9550-18-1-0148 and FA9550-21-1-0008). B.F additionally acknowledges support from the National Science Foundation under Grant CCF-2044923 (CAREER). We would like to thank Dieter van Melkebeek, as well as the anonymous reviewers, for many helpful comments on a preliminary draft.
	
	\bibliographystyle{plainurl}
	\let\OLDthebibliography\thebibliography
	\renewcommand\thebibliography[1]{
		\OLDthebibliography{#1}
		\setlength{\parskip}{0pt}
		\setlength{\itemsep}{0pt plus 0.3ex}
	}
	\bibliography{references} 
	
	\appendix
	
	\section{A TM-based Proof of BPL \texorpdfstring{$\subseteq$}{in} B\texorpdfstring{Q\textsubscript{U}}{QU}L=BQL}\label{sec:appendix:turingMachineVersionOfBqlEqualsBqul}
	
	While, trivially, $\mathsf{BPL} \subseteq \mathsf{BQL}$, it is not obvious, \textit{a priori}, that $\mathsf{BPL} \subseteq \mathsf{BQ_U L}$. To the best of our knowledge, the strongest partial result in this direction is the classic result of Watrous \cite[Theorem 4.12]{watrous1999space}, which showed that $\mathsf{BPL}$ is contained in a variant of $\mathsf{BQ_U L}$ in which there is no bound on the running time of the QTM. By \Cref{thm:bqulComplete}, $poly\text{-conditioned-}\mathsf{MATPOW} \in \mathsf{BQ_U L}$. As we next observe, this implies $\mathsf{BPL} \subseteq \mathsf{BQ_U L}$ and, more strongly, $\mathsf{BQL} = \mathsf{BQ_U L}$. Of course, the statement $\mathsf{BQL} = \mathsf{BQ_U L}$ immediately implies $\mathsf{BPL} \subseteq \mathsf{BQ_U L}$; nevertheless, we will first show, directly, that $\mathsf{BPL} \subseteq \mathsf{BQ_U L}$. 
	
	\begin{proposition}\label{prop:bplInBqul}
		$\mathsf{BPL} \subseteq \mathsf{BQ_U L}$.
	\end{proposition}
	\begin{proof}
		Suppose $\mathcal{P}=(\mathcal{P}_1,\mathcal{P}_0) \in \mathsf{BPL}$. Then there is some probabilistic TM $M$ that recognizes $\mathcal{P}$ with two-sided error $\leq \frac{1}{3}$ within time $t(n)=n^{O(1)}$ and space $s(n)=O(\log n)$, for any input $w \in \mathcal{P}$ of length $n=\lvert w \rvert$. Let $\lvert M \rvert$ denote the size of the finite control of $M$, let $\Gamma$ denote the work-tape alphabet of $M$, and let $c(n)=\lvert M \rvert (n+2) (s(n)) \lvert \Gamma \rvert^{s(n)}=n^{O(1)}$ denote the number of possible configurations of $M$ on inputs of length $n$. It is well-known that, for input $w \in \mathcal{P}$, one may construct, in deterministic space $O(\log(\lvert w \rvert))$ a stochastic matrix $A_w \in \widehat{\mat}(c(n))$ and values $x_w,y_w \in [c(n)]$ such that $A_w^t[x_w,y_w]$ is precisely the probability that $M$ accepts $w$ within $t$ steps \cite{nisan1992pseudorandom,doron2017approximating}; this implies $\mathsf{MATPOW}(\langle A_w,x_w,y_w,\frac{2}{3} \rangle)= \mathcal{P}(w)$. Note that, as $A_w$ is stochastic, so is $A_w^t$, $\forall t \in \mathbb{N}$; this implies $\sigma_1(A_w^t) \leq \sqrt{c(n)}=n^{O(1)}$, which then implies $\langle A_w,x_w,y_w,\frac{2}{3} \rangle \in \mathsf{MATPOW}_{c(n),t(n),\sqrt{c(n)},3}$. By \Cref{thm:bqulComplete}, $\mathsf{MATPOW}_{c(n),t(n),\sqrt{c(n)},3} \in \mathsf{BQ_U L}$, which implies $\mathcal{P} \in \mathsf{BQ_U L}$.
	\end{proof}
	
	By applying an analogous argument to general quantum Turing machines (where the stochastic matrix that describes a single step of the computation of a probabilistic TM is replaced by the quantum channel that describes a single step of the computation of a quantum TM), we may then show that $\mathsf{BQL} \subseteq \mathsf{BQ_U L}$ (and, therefore, that $\mathsf{BQL} = \mathsf{BQ_U L}$). Here, $\mathsf{BQL}$ is defined in terms of a logspace \textit{quantum Turing machine} (QTM), as was the case in, for instance \cite{watrous1999space,watrous2001quantum,watrous2003complexity,watrous2009encyclopedia,ta2013inverting,melkebeek2012time,perdrix2006classically,jozsa2010matchgate}, rather than the equivalent model of a uniform family of \textit{general quantum circuits} used in this paper. 
	
	For concreteness, we use the \textit{classically controlled} logspace (general) QTM defined by Watrous \cite{watrous2003complexity} (with the minor alteration that we require all transition amplitudes of the QTM to be computable in $\mathsf{L}$); however, we note that our result would apply equally well to any ``reasonable" logspace QTM model that is classically controlled (this includes all models considered in all of the papers cited above). In brief, such a QTM $M$ consists of a (classical) finite control, an internal quantum register of constant size, a classical ``measurement" register of constant size, and three tapes: (1) a read-only input tape that, on any input $w$, contains the string $\#_L w \#_R$, where $\#_L$ and $\#_R$ are special symbols that serve as left and right end-markers, (2) a read/write classical work tape consisting of $s(\lvert w \rvert)=O(\log \lvert w \rvert)$ cells, each of which holds a symbol from some finite alphabet $\Gamma$, and (3) a read/write quantum work tape, consisting of $s(\lvert w \rvert)=O(\log \lvert w \rvert)$ qubits. Each of the tapes has a single bidirectional head. At the start of the computation, both work-tapes are ``blank" (to be precise, each cell of the classical work tape contains some specified blank-symbol in $\Gamma$ and each qubit of the quantum work tape is in the state $\ket{0}$); each qubit of the internal quantum register is also in the state $\ket{0}$. Each step of the computation of $M$ involves applying a \textit{selective quantum operation} to the combined register consisting of the internal quantum register and the single qubit that is currently under the head of the quantum work tape; the particular choice of which selective quantum operation to perform may depend on the state of the finite control and the symbols currently under the heads of the input tape and classical work tape. The (classical) result of this quantum operation is stored in the measurement register. Then, depending on this result, as well as on the state of the finite control and the symbols currently under the heads of the input tape and classical work tape, the classical configuration of the machine evolves; to be precise, the state of the finite control is updated, a symbol is written on the classical work-tape, and the head of each work tape moves up to one cell in either direction. The machine accepts (resp.) rejects its input by entering a special (classical) accepting (resp. rejecting state). See \cite{watrous2003complexity} for a complete definition.
	
	\begin{proposition}\label{prop:turingMachineVersionOfBqlEqualsBqul}
		$\mathsf{BQL}=\mathsf{BQ_U L}$. 
	\end{proposition}
	\begin{proof}
		Trivially, $\mathsf{BQL} \supseteq \mathsf{BQ_U L}$. We next show $\mathsf{BQL} \subseteq \mathsf{BQ_U L}$. Suppose $\mathcal{P}=(\mathcal{P}_1,\mathcal{P}_0) \in \mathsf{BQL}$. By definition, there is some QTM $M$ such that the following conditions are satisfied: (1) on any input $w \in \mathcal{P}$ of length $n=\lvert w \rvert$, $M$ runs in space at most $s(n)=O(\log n)$ (and hence time $t(n)=2^{O(s(n))}$), (2) if $w \in \mathcal{P}_1$, then $\Pr[M \text{ accepts } w]\geq \frac{2}{3}$, and (3) if $w \in \mathcal{P}_0$, then $\Pr[M \text{ accepts } w]\leq \frac{1}{3}$. 
		
		Consider running $M$ on some input of length $n$. At any particular point in time, the configuration of (a single probabilistic branch of) $M$ consists of the current (classical) state of the finite control, the (quantum) contents of the internal quantum register, the (classical) contents of the measurement register, the (classical) positions of the heads on the read only input-tape and the classical and quantum work-tapes, the current (classical) contents of the classical work-tape, and the current (quantum) contents of the quantum work-tape. Let $\lvert M \rvert$ denote the size of the finite control, let $b_m$ denote the number of bits of the measurement register, let $b_q$ denote the number of qubits of the internal quantum register, and let $\Gamma$ denote the classical work-tape alphabet. Let $C_n$ denote the set of all possible classical configurations of $M$ on inputs of length $n$, where $\lvert C_n \rvert=\lvert M \rvert 2^{b_m} (n+2) s(n)^2 \lvert \Gamma \rvert^{s(n)}=n^{O(1)}$. Each classical configuration $c \in C_n$ corresponds to the element $\ket{c}$ in the natural orthonormal basis of the Hilbert space $\mathbb{C}^{C_n}$. Let $Q_n$ denote the set of $\lvert Q_n \rvert=2^{s(n)+b_q}=n^{O(1)}$ quantum basis states corresponding to the quantum work-tape and internal quantum register. The contents of the quantum work-tape and the internal quantum register is then described by some $\ket{\psi} \in \mathbb{C}^{Q_n}$. Then each configuration of $M$ on an input of length $n$ corresponds to an element $\ket{c}\ket{\psi}$ of the Hilbert space $\mathcal{H}_{M,n}=\mathbb{C}^{C_n} \otimes \mathbb{C}^{Q_n}$. Let $d(n)=\dim(\mathcal{H}_{M,n})=\lvert C_n \rvert \lvert Q_n \rvert =n^{O(1)}$.
		
		Consider some input $w \in \mathcal{P}$. Let $n=\lvert w \rvert$ denote the length of $w$, let $\Phi_{M,w} \in \chan(\mathcal{H}_{M,n})$ denote the quantum channel that corresponds to a single step of the computation of $M$ on $w$, and let $K(\Phi_{M,w}) \in \widehat{\mat}(d^2(n))$ denote the natural representation of $\Phi_{M,w}$. For any $t \in \mathbb{N}$, we have $\Phi_{M,w}^t \in \chan(\mathcal{H}_{M,n})$, which implies $\sigma_1((K(\Phi_{M,w}))^t)=\sigma_1(K(\Phi_{M,w}^t))\leq \sqrt{d(n)}=n^{O(1)}$ \cite[Theorem 1]{roga2013entropic}. Let $\ket{\psi_{start}^n} =\ket{c_{start}^n}\ket{q_{start}^n} \in \mathcal{H}_{M,n}$ denote the starting configuration of $M$ on an input of length $n$, where $c_{start}^n \in C_n$ is the classical part of the starting configuration, and $\ket{q_{start}^n}=\ket{0^{s(n)+b_q}} \in \mathbb{C}^{Q_n}$ is the quantum part. Without loss of generality we may, for convenience, assume that $M$ ``cleans-up" its workspace at the end of the computation, by returning both its classical and quantum work tapes to the ``blank" configuration described above; in particular, this implies that $M$ has a unique accepting configuration $\ket{\psi_{accept}^n} =\ket{c_{accept}^n}\ket{q_{start}^n} \in \mathcal{H}_{M,n}$ on any input of length $n$. Let $A_w=K(\Phi_{M,w}) \in \widehat{\mat}(d^2(n))$, $x_w=\text{vec}(\ket{\psi_{accept}^n}\bra{\psi_{accept}^n}) \in [d^2(n)]$, and $y_w=\text{vec}(\ket{\psi_{start}^n}\bra{\psi_{start}^n}) \in [d^2(n)]$. Then $A_w^t[x_w,y_w]$ is precisely the probability that $M$ accepts $w$ within $t$ steps. Thus, $\langle A_w,x_w,y_w,\frac{2}{3} \rangle \in \mathsf{MATPOW}_{d^2(n),t(n),\sqrt{d(n)},3}$ and $\mathsf{MATPOW}(\langle A_w,x_w,y_w,\frac{2}{3} \rangle)= \mathcal{P}(w)$. Finally, by \Cref{thm:bqulComplete}, $\mathsf{MATPOW}_{d^2(n),t(n),\sqrt{d(n)},3} \in \mathsf{BQ_U L}$, which implies $\mathcal{P} \in \mathsf{BQ_U L}$.
	\end{proof}	
	
	\section{Proof of Theorem 1}\label{sec:appendix:proofOfTheorem1}
	
	\begin{proof}[Proof of \Cref{thm:intro:bquspaceEqualsBqspace}]
		Clearly, $\mathsf{BQ_U SPACE}(s(n)) \subseteq \mathsf{BQSPACE}(s(n))$. The containment $\mathsf{BQSPACE}(s(n)) \subseteq \mathsf{BQ_U SPACE}(s(n))$ follows from \Cref{thm:bqlEqualsBqulEqualsQmal} and a standard padding argument, which we now state. 	Suppose $\mathcal{P}=(\mathcal{P}_1,\mathcal{P}_0) \in \mathsf{BQSPACE}(s(n))$. By definition, there is some $\mathsf{DSPACE}(s(n))$-uniform family of general quantum circuits $\{\Phi_w=(\Phi_{w,1},\ldots,\Phi_{w,t_w}):w \in \mathcal{P}\}$, where $\Phi_w$ acts on $h_w=O(s(\lvert w \rvert))$ qubits and has $t_w=2^{O(s(\lvert w \rvert))}$ gates, such that if $w \in \mathcal{P}_1$, then $\Pr[\Phi_w \text{ accepts } w]\geq \frac{2}{3}$, and if $w \in \mathcal{P}_0$, then $\Pr[\Phi_w \text{ accepts } w]\leq \frac{1}{3}$. Let $M$ denote a deterministic Turing machine (DTM) that produces this family of circuits within the stated space bound. Let $\Sigma$ denote the finite alphabet over which $\mathcal{P}$ is defined, and assume, without loss of generality, that $\{0,1\} \subseteq \Sigma$.
		
		We define $\mathcal{P}^{\log}=(\mathcal{P}_1^{\log},\mathcal{P}_0^{\log}) \subseteq \Sigma^*$ such that $\mathcal{P}_j^{\log}=\{w01^{2^{s(\lvert w \rvert)}}:w \in \mathcal{P}_j\}$, for $j \in \{0,1\}$. We next show that $\mathcal{P}^{\log} \in \mathsf{BQL}$, by exhibiting a family of general quantum circuits $\{\Phi_x^{\log}=(\Phi_{x,1}^{\log},\ldots,\Phi_{x,t_x^{\log}}^{\log}):x \in \mathcal{P}^{\log}\}$, with the appropriate parameters, that recognizes $\mathcal{P}^{\log}$. Begin by noticing that, as $s$ is space-constructible, there is a DTM $D$ that uses space $O(\log n)$ on all inputs of length $n$, where, on any input $x \in \Sigma^*$, $D$ checks if $x=w01^{2^{s(\lvert w \rvert)}}$, for some $w \in \Sigma^*$. If $x$ is of this form, then $D$ marks the rightmost symbol of $w$; otherwise, $D$ rejects. We then construct a DTM $M'$ which, on input $x \in \Sigma^*$ produces $\Phi_x^{\log}$, as follows. First, $M'$ runs $D$. If $D$ rejects, then $M'$ outputs a trivial single-gate circuit, which acts on a single qubit and always rejects. Otherwise (i.e., when the input $x$ is of the form $w01^{2^{s(\lvert w \rvert)}}$), $M'$ simulates $M$ on the prefix $w$, producing the circuit $\Phi_w$; note that, in this case, $\lvert x \rvert=\lvert w\rvert+1+2^{s(\lvert w \rvert)}$, which implies $M'$ runs in space $O(s(\lvert w \rvert))=O(\log(2^{s(\lvert w \rvert)}))=O(\log(\lvert x \rvert))$, and that $\Phi_x^{\log}$ acts on $h_x^{\log}=h_w=O(s(\lvert w \rvert))=O(\log(\lvert x \rvert))$ qubits and has $t_x^{\log}=t_w=2^{O(s(\lvert w \rvert))}=\lvert x \rvert^{O(1)}$ gates. Therefore, $\{\Phi_x^{\log}:x \in \mathcal{P}^{\log}\}$ is a $\mathsf{L}$-uniform family of general quantum circuits, with the appropriate parameters, that recognizes $\mathcal{P}^{\log}$ with two-sided bounded-error $\frac{1}{3}$.
		
		Thus, by \Cref{thm:bqlEqualsBqulEqualsQmal}, $\mathcal{P}^{\log} \in \mathsf{BQ_U L}$. Let $\{Q_x^{\log}:x \in \mathcal{P}^{\log}\}$ denote a $\mathsf{L}$-uniform family of (unitary) quantum circuits that recognizes $\mathcal{P}^{\log}$ with two-sided bounded-error $\frac{1}{3}$, where $Q_x^{\log}$ acts on $O(\log(\lvert x \rvert))$ qubits and has $\lvert x \rvert^{O(1)}$ gates, and let $M_U$ denote a logspace DTM that produces this circuit family. We then define a DTM $M_U'$, which, on input $w \in \Sigma^*$, simply simulates $M_U$ on $x=w01^{2^{s(\lvert w \rvert)}}$; in particular, in order to keep track of the simulated head of $M_U$ when it is in the suffix $01^{2^{s(\lvert w \rvert)}}$, $M_U'$ marks $s(\lvert w \rvert)+1$ cells on its work tape (recall that $s$ is space-constructible), which is then used as a classical counter that can count up to $2^{s(\lvert w \rvert)+1}-1$. By an analysis similar to that of the previous paragraph, we see that $M_U'$ runs in space $O(s(n))$ and that the circuit family $\{Q_w:w \in \mathcal{P}\}$ that it produces recognizes $\mathcal{P}$ and has the correct parameters. 
	\end{proof}	

	\section{Proof of Lemma 34}\label{sec:appendix:proofOfRqumalInRqul}
	
	\begin{proof}[Proof of \Cref{lemma:rqmalInRqul}]
		Suppose $\mathcal{P}=(\mathcal{P}_1,\mathcal{P}_0) \in \mathsf{RQ_U MAL}$. By definition, there is a $\mathsf{L}$-uniform family of (unitary) quantum circuits $\{V_w:w \in \mathcal{P}\}$, where $V_w$ acts on $m_w+h_w=O(\log \lvert w \rvert)$ qubits and has $t_w=poly(\lvert w \rvert)$ gates, such that  $w \in \mathcal{P}_1 \Rightarrow \exists \ket{\psi} \in \Psi_{m_w}, \Pr[V_w \text{ accepts } w,\ket{\psi}] \geq c= \frac{1}{2}$, and $w \in \mathcal{P}_0 \Rightarrow \forall \ket{\psi} \in \Psi_{m_w}, \Pr[V_w \text{ accepts } w,\ket{\psi}]= k=0$, where $\Pr[V_w \text{ accepts } w,\ket{\psi}]= \lVert \Pi_1 V_w (\ket{\psi} \otimes \ket{0^{h_w}}) \rVert^2$. 
		
		As in the proof of \cite[Theorem 3.8]{marriott2005quantum}, let $A_w=(I_{2^{m_w}} \otimes \bra{0^{h_w}})V_w^{\dagger} \Pi_1 V_w (I_{2^{m_w}} \otimes \ket{0^{h_w}}) \in \pos(2^{m_w})$; then $w \in \mathcal{P}_1 \Rightarrow \tr(A_w) \geq c=\frac{1}{2}$ and $w \in \mathcal{P}_0 \Rightarrow \tr(A_w) \leq 2^{m_w} k=0 \Rightarrow \tr(A_w)=0$. Similar to the proof of \cite[Theorem 14]{kobayashi2003quantum} (cf. \cite[Theorem 3.10]{marriott2005quantum}), we define a $\mathsf{L}$-uniform family of (unitary) quantum circuits $\{Q_w:w \in \mathcal{P}\}$, where $Q_w$ acts on $2m_w+h_w=O(\log \lvert w \rvert)$ qubits and has $t_w+O(m_w)=poly(\lvert w \rvert)$ gates, such that, when $Q_w$ is applied to the state $\ket{0^{2m_w+h_w}}$, it simulates $V_w$ on $\ket{q} \otimes \ket{0^{h_w}}$, where $\ket{q} \in \Psi_{m_w}$ is drawn uniformly at random from the $2^{m_w}$ standard basis elements of $\Psi_{m_w}$. We have $\Pr[Q_w \text{ accepts } w]=\tr(A_w 2^{-m_w} I_{2^{m_w}})=2^{-m_w}\tr(A_w)$; thus, $$w \in \mathcal{P}_1 \Rightarrow \Pr[Q_w \text{ accepts } w]=2^{-m_w}\tr(A_w)\geq 2^{-(m_w+1)}=1/poly(\lvert w \rvert),$$ 
		$$\text{and } w \in \mathcal{P}_0 \Rightarrow \Pr[Q_w \text{ accepts } w]=2^{-m_w}\tr(A_w)=0.$$ 
		
		Therefore, $\mathcal{P} \in \mathsf{Q_U SPACE}(\log n)_{\frac{1}{poly(n)},0}$. By \cite[Lemma 5.1]{watrous2001quantum} $\mathsf{Q_U SPACE}(\log n)_{\frac{1}{poly(n)},0}=\mathsf{RQ_U L}$, which implies $\mathcal{P} \in \mathsf{RQ_U L}$.
	\end{proof}
\end{document}